\documentclass[10pt, conference]{IEEEtran}

\usepackage[utf8]{inputenc}

\usepackage[sort&compress, numbers]{natbib}

\usepackage[final]{graphicx}
\usepackage{tkz-graph}
\usepackage{tikz}
\usepackage{tikz-cd}
\graphicspath{{../figures/}}
\usepackage{caption}
\usepackage{subcaption}
\usepackage{stackengine}
\usepackage{adjustbox}
\usepackage{url}
\usepackage{amssymb,mathtools,nccmath,amsthm,mathrsfs}
\usepackage{accents}
\theoremstyle{definition}
\newtheorem{notation}{Notation}
\newtheorem{definition}{Definition}
\newtheorem{theorem}{Theorem}

\newtheorem{remark}{Remark}
\newtheorem{example}{Example}

\DeclareFontFamily{OT1}{pzc}{}
\DeclareFontShape{OT1}{pzc}{m}{it}{<-> s * [1.10] pzcmi7t}{}
\DeclareMathAlphabet{\mathpzc}{OT1}{pzc}{m}{it}

\usepackage{hyperref}
\hypersetup{
    colorlinks=true,
    linkcolor=black,
    citecolor=black,
    filecolor=black,
    urlcolor=black,
}

\usepackage[letterspace=150]{microtype}

\tikzstyle{arrow-label}=[above=-1,fill=none,scale=0.8]
\tikzstyle{my-vertex-label}=[scale=0.8]
\tikzstyle{quiver-vertex}=[minimum size=1pt, inner sep=1pt, opacity=1, color=black]
\tikzstyle{category-vertex}=[minimum size=1pt, inner sep=1pt, opacity=1, circle, shading=ball, ball color=black!80!white]

\tikzstyle{rule-vertex}=[minimum size=0pt, inner sep=0pt]
\tikzstyle{rule-arrow}=[-{Implies},double]

\tikzstyle{dead-organism}=[densely dotted,-{Stealth[scale=0.8]},line width=0.65, opacity=0.45]
\tikzstyle{alive-organism}=[-{Stealth[scale=0.8]},line width=0.65]
\tikzstyle{hide-organism}=[draw=none]
\tikzstyle{composite-organism}=[dotted,-{Stealth[scale=0.8]},line width=0.65]

\tikzstyle{hide-label}=[opacity=0]
\tikzstyle{alive-label}=[above=-1,fill=none,scale=0.8]

\begin{document}

\newcommand{\getvertex}[6]
{
	\ifthenelse{\equal{#1}{show}}{\Vertex[L=\hbox{#2}, style={my-vertex-label}, Lpos={#3}, x=#4, y=#5]{#6}}
	{\Vertex[L=\hbox{#2}, style={my-vertex-label}, Lpos={#3}, x=#4, y=#5, empty=true]{#6}}
	
}

\newcommand{\machineruleunary}[4]
{
	\GraphInit[vstyle=Classic]
	\tikzset{VertexStyle/.append style=quiver-vertex}
	
	\Vertex[L=\hbox{}, x=#1, y=#2]{b00}
	\Vertex[L=\hbox{}, x=#1+1, y=#2]{b01}
	\Edge[style=#3-organism](b00)(b01)
	
	\draw[rule-arrow] (#1+0.5,#2-0.1)--(#1+0.5,#2-0.5);
	
	\Vertex[L=\hbox{}, x=#1, y=#2-0.6]{a00}
	\Vertex[L=\hbox{}, x=#1+1, y=#2-0.6]{a01}
	\Edge[labelstyle=below,style=#4-organism](a00)(a01)
	
	\draw (#1-0.2,#2-0.8) rectangle (#1+1.2,#2+0.2);
}

\newcommand{\machinerulebinary}[6]
{
	\GraphInit[vstyle=Classic]
	\tikzset{VertexStyle/.append style=quiver-vertex}
	
	\Vertex[L=\hbox{}, x=#1, y=#2]{b00}
	\Vertex[L=\hbox{}, x=#1+1, y=#2]{b01}
	\Vertex[L=\hbox{}, x=#1+2, y=#2]{b02}
	\Edge[style=#3-organism](b00)(b01)
	\Edge[style=#4-organism](b01)(b02)

	\ifthenelse{\equal{#6}{right}}{  		
  		\draw[rule-arrow] (#1+1.5,#2-0.1)--(#1+1.5,#2-0.5);
  		\Vertex[L=\hbox{}, x=#1+1, y=#2-0.6]{a00}
		\Vertex[L=\hbox{}, x=#1+2, y=#2-0.6]{a01}
		\Edge[labelstyle=below,style=#5-organism](a00)(a01)
	}{
		\draw[rule-arrow] (#1+0.5,#2-0.1)--(#1+0.5,#2-0.5);
		\Vertex[L=\hbox{}, x=#1, y=#2-0.6]{a00}
		\Vertex[L=\hbox{}, x=#1+1, y=#2-0.6]{a01}
		\Edge[labelstyle=below,style=#5-organism](a00)(a01)  		
	}
	\draw (#1-0.2,#2-0.8) rectangle (#1+2.2,#2+0.2);
}

\newcommand{\machineruleternary}[6]
{
	\GraphInit[vstyle=Classic]
	\tikzset{VertexStyle/.append style=quiver-vertex}
	
	\Vertex[L=\hbox{}, x=#1, y=#2]{b00}
	\Vertex[L=\hbox{}, x=#1+1, y=#2]{b01}
	\Vertex[L=\hbox{}, x=#1+2, y=#2]{b02}
	\Vertex[L=\hbox{}, x=#1+3, y=#2]{b03}
	\Edge[style=#3-organism](b00)(b01)
	\Edge[style=#4-organism](b01)(b02)
	\Edge[style=#5-organism](b02)(b03)
	
	\draw[rule-arrow] (#1+1.5,#2-0.1)--(#1+1.5,#2-0.5);
	
	\Vertex[L=\hbox{}, x=#1+1, y=#2-0.6]{a00}
	\Vertex[L=\hbox{}, x=#1+2, y=#2-0.6]{a01}
	\Edge[labelstyle=below,style=#6-organism](a00)(a01)
	
	\draw (#1-0.2,#2-0.8) rectangle (#1+3.2,#2+0.2);
}

\newcommand\quiverx[6]{%
    \def\xa{#1}%
    \def\xb{#2}%
    \def\xc{#3}%
    \def\xd{#4}%
    \def\xe{#5}%
    \def\xf{#6}%
}
\newcommand\quivery[1]{%
    \def\ya{#1}%
}
\newcommand\quiverz[3]{%
    \def\za{#1}%
    \def\zb{#2}%
    \def\zc{#3}%
}
\newcommand{\machineexample}[4]
{
	\GraphInit[vstyle=Classic]
	\tikzset{VertexStyle/.append style=quiver-vertex}	
	
	\node[text width=6.5cm] at (#1+2.85,#2+0.9) {#4};
	\ifthenelse{#3 > -1}{
		\node[text width=1cm] at (#1-0.9,#2-1) {t=#3};
	}{}		

	\Vertex[L=\hbox{$x_1$}, style={my-vertex-label}, Lpos={left=.04cm}, x=#1, y=#2]{x0}
	\Vertex[L=\hbox{$x_2$}, style={my-vertex-label}, Lpos={below=.04cm}, x=#1+1, y=#2-0.2]{x1}
	\Vertex[L=\hbox{$x_3$}, style={my-vertex-label}, Lpos={below=.04cm}, x=#1+2, y=#2]{x2}
	\Vertex[L=\hbox{$x_4$}, style={my-vertex-label}, Lpos={below=.04cm}, x=#1+3, y=#2]{x3}
	\Vertex[L=\hbox{$x_5$}, style={my-vertex-label}, Lpos={below=.04cm}, x=#1+4, y=#2]{x4}
	\Vertex[L=\hbox{$x_6$}, style={my-vertex-label}, Lpos={below=.04cm}, x=#1+5, y=#2-0.2]{x5}
	\Vertex[L=\hbox{$x_7$}, style={my-vertex-label}, Lpos={right=.04cm}, x=#1+6, y=#2]{x6}
	\tikzset{EdgeStyle/.style=\xa-organism}\Edge[labelstyle={arrow-label},label=$\alpha_1$](x0)(x1)
	\tikzset{EdgeStyle/.style=\xb-organism}\Edge[labelstyle={arrow-label},label=$\alpha_2$](x1)(x2)
	\tikzset{EdgeStyle/.style=\xc-organism}\Edge[labelstyle={arrow-label},label=$\alpha_3$](x2)(x3)
	\tikzset{EdgeStyle/.style=\xd-organism}\Edge[labelstyle={arrow-label},label=$\alpha_4$](x3)(x4)
	\tikzset{EdgeStyle/.style=\xe-organism}\Edge[labelstyle={arrow-label},label=$\alpha_5$](x4)(x5)
	\tikzset{EdgeStyle/.style=\xf-organism}\Edge[labelstyle={arrow-label},label=$\alpha_6$](x5)(x6)
	
	\Vertex[L=\hbox{$x_8$}, style={my-vertex-label}, Lpos={left=.04cm}, x=#1+2.5, y=#2-1.6]{y0}
	\Vertex[L=\hbox{$x_9$}, style={my-vertex-label}, Lpos={right=.04cm}, x=#1+3.5, y=#2-1.6]{y1}
	\tikzset{EdgeStyle/.style=\ya-organism}\Edge[labelstyle={arrow-label},label=$\alpha_7$](y0)(y1)
	
	\Vertex[L=\hbox{$x_{10}$}, style={my-vertex-label}, Lpos={left=.04cm}, x=#1+1.5, y=#2-2.4]{z0}
	\Vertex[L={$x_{11}$}, style={my-vertex-label}, Lpos={below=.04cm}, x=#1+2.5, y=#2-2.4]{z1}
	\Vertex[L=\hbox{$x_{12}$}, style={my-vertex-label}, Lpos={below=.04cm}, x=#1+3.5, y=#2-2.4]{z2}
	\Vertex[L=\hbox{$x_{13}$}, style={my-vertex-label}, Lpos={right=.04cm}, x=#1+4.5, y=#2-2.4]{z3}
	\tikzset{EdgeStyle/.style=\za-organism}\Edge[labelstyle={arrow-label},label=$\alpha_8$](z0)(z1)
	\tikzset{EdgeStyle/.style=\zb-organism}\Edge[labelstyle={arrow-label},label=$\alpha_9$](z1)(z2)
	\tikzset{EdgeStyle/.style=\zc-organism}\Edge[labelstyle={arrow-label},label=$\alpha_{10}$](z2)(z3)
		
	\draw (#1-0.5,#2-3.2) rectangle (#1+6.5,#2+1.2);
}

\newcommand{\spaceexample}[3]
{
	\GraphInit[vstyle=Classic]
	\tikzset{VertexStyle/.append style=category-vertex}
	
	\ifthenelse{#3 > -1}{\node[text width=3cm] at (#1+1.15,#2+0.9){$(R \circ T \circ P)(\overrightarrow{Q}^#3)$};}{}

	\ifthenelse{\equal{\xa}{alive}}{\getvertex{show}{$d_1$}{left=.04cm}{#1}{#2}{x0}}{\getvertex{hide}{$d_1$}{left=.04cm}{#1}{#2}{x0}}
	\ifthenelse{\equal{\xa}{alive} \OR \equal{\xb}{alive}}{\getvertex{show}{$d_2$}{below=0.04cm}{#1+1}{#2-0.2}{x1}}{\getvertex{hide}{$d_2$}{below=0.04cm}{#1+1}{#2-0.2}{x1}}
	\ifthenelse{\equal{\xb}{alive} \OR \equal{\xc}{alive}}{\getvertex{show}{$d_3$}{below=.04cm}{#1+2}{#2}{x2}}{\getvertex{hide}{$d_3$}{below=.04cm}{#1+2}{#2}{x2}}
	\ifthenelse{\equal{\xc}{alive} \OR \equal{\xd}{alive}}{\getvertex{show}{$d_4$}{below=.04cm}{#1+3}{#2}{x3}}{\getvertex{hide}{$d_4$}{below=.04cm}{#1+3}{#2}{x3}}
	\ifthenelse{\equal{\xd}{alive} \OR \equal{\xe}{alive}}{\getvertex{show}{$d_5$}{below=.04cm}{#1+4}{#2}{x4}}{\getvertex{hide}{$d_5$}{below=.04cm}{#1+4}{#2}{x4}}
	\ifthenelse{\equal{\xe}{alive} \OR \equal{\xf}{alive}}{\getvertex{show}{$d_6$}{below=.04cm}{#1+5}{#2-0.2}{x5}}{\getvertex{hide}{$d_6$}{below=.04cm}{#1+5}{#2-0.2}{x5}}
	\ifthenelse{\equal{\xf}{alive}}{\getvertex{show}{$d_7$}{right=.04cm}{#1+6}{#2}{x6}}{\getvertex{hide}{$d_7$}{right=.04cm}{#1+6}{#2}{x6}}
		
	\tikzset{EdgeStyle/.style=\xa-organism}\Edge[label=$f_1$,labelstyle=\xa-label](x0)(x1)
	\tikzset{EdgeStyle/.style=\xb-organism}\Edge[label=$f_2$,labelstyle=\xb-label](x1)(x2)
	\tikzset{EdgeStyle/.style=\xc-organism}\Edge[label=$f_3$,labelstyle=\xc-label](x2)(x3)
	\tikzset{EdgeStyle/.style=\xd-organism}\Edge[label=$f_4$,labelstyle=\xd-label](x3)(x4)
	\tikzset{EdgeStyle/.style=\xe-organism}\Edge[label=$f_5$,labelstyle=\xe-label](x4)(x5)
	\tikzset{EdgeStyle/.style=\xf-organism}\Edge[label=$f_6$,labelstyle=\xf-label](x5)(x6)
		
	\ifthenelse{\equal{\xa}{alive} \AND \equal{\xb}{alive}}{
		\tikzset{EdgeStyle/.style=composite-organism}\Edge[style={bend right=50}, labelstyle=below](x0)(x2)	
		\ifthenelse{\equal{\xc}{alive}}{
			\tikzset{EdgeStyle/.style=composite-organism}\Edge[style={bend left=50}, labelstyle=below](x0)(x3)
			\ifthenelse{\equal{\xd}{alive}}{
				\tikzset{EdgeStyle/.style=composite-organism}\Edge[style={bend right=50}, labelstyle=below](x0)(x4)
				\ifthenelse{\equal{\xe}{alive}}{
					\tikzset{EdgeStyle/.style=composite-organism}\Edge[style={bend left=50}, labelstyle=below](x0)(x5)
					\ifthenelse{\equal{\xf}{alive}}{
						\tikzset{EdgeStyle/.style=composite-organism}\Edge[style={bend right=50}, labelstyle=below](x0)(x6)
					}{}
				}{}
			}{}
		}{}		
	}{}
	
	\ifthenelse{\equal{\xb}{alive} \AND \equal{\xc}{alive}}{
		\tikzset{EdgeStyle/.style=composite-organism}\Edge[style={bend right=50}, labelstyle=below](x1)(x3)	
		\ifthenelse{\equal{\xd}{alive}}{
			\tikzset{EdgeStyle/.style=composite-organism}\Edge[style={bend left=50}, labelstyle=below](x1)(x4)
			\ifthenelse{\equal{\xe}{alive}}{
				\tikzset{EdgeStyle/.style=composite-organism}\Edge[style={bend right=50}, labelstyle=below](x1)(x5)
				\ifthenelse{\equal{\xf}{alive}}{
					\tikzset{EdgeStyle/.style=composite-organism}\Edge[style={bend left=50}, labelstyle=below](x1)(x6)
				}{}
			}{}
		}{}		
	}{}
	
	\ifthenelse{\equal{\xc}{alive} \AND \equal{\xd}{alive}}{
		\tikzset{EdgeStyle/.style=composite-organism}\Edge[style={bend right=50}, labelstyle=below](x2)(x4)	
		\ifthenelse{\equal{\xe}{alive}}{
			\tikzset{EdgeStyle/.style=composite-organism}\Edge[style={bend left=50}, labelstyle=below](x2)(x5)
			\ifthenelse{\equal{\xf}{alive}}{
				\tikzset{EdgeStyle/.style=composite-organism}\Edge[style={bend right=50}, labelstyle=below](x2)(x6)
			}{}
		}{}		
	}{}
	
	\ifthenelse{\equal{\xd}{alive} \AND \equal{\xe}{alive}}{
		\tikzset{EdgeStyle/.style=composite-organism}\Edge[style={bend right=50}, labelstyle=below](x3)(x5)	
		\ifthenelse{\equal{\xf}{alive}}{
			\tikzset{EdgeStyle/.style=composite-organism}\Edge[style={bend left=50}, labelstyle=below](x3)(x6)
		}{}		
	}{}	

	\ifthenelse{\equal{\xe}{alive} \AND \equal{\xf}{alive}}{
		\tikzset{EdgeStyle/.style=composite-organism}\Edge[style={bend right=50}, labelstyle=below](x4)(x6)
	}{}
		
	\ifthenelse{\equal{\ya}{alive}}{
		\getvertex{show}{$d_8$}{left=0.04cm}{#1+2.5}{#2-1.6}{y0}
		\getvertex{show}{$d_9$}{right=0.04cm}{#1+3.5}{#2-1.6}{y1}
		\tikzset{EdgeStyle/.style=\ya-organism}\Edge[label=$f_7$,labelstyle=\ya-label](y0)(y1)
	}{}
	
	\ifthenelse{\equal{\za}{alive}}{\getvertex{show}{$d_{10}$}{left=0.04cm}{#1+1.5}{#2-2.4}{z0}}{\getvertex{hide}{$d_{10}$}{left=0.04cm}{#1+1.5}{#2-2.4}{z0}}
	\ifthenelse{\equal{\za}{alive} \OR \equal{\zb}{alive}}{\getvertex{show}{$d_{11}$}{below=0.04cm}{#1+2.5}{#2-2.4}{z1}}{\getvertex{hide}{$d_{11}$}{below=0.04cm}{#1+2.5}{#2-2.4}{z1}}
	\ifthenelse{\equal{\zb}{alive} \OR \equal{\zc}{alive}}{\getvertex{show}{$d_{12}$}{below=0.04cm}{#1+3.5}{#2-2.4}{z2}}{\getvertex{hide}{$d_{12}$}{below=0.04cm}{#1+3.5}{#2-2.4}{z2}}
	\ifthenelse{\equal{\zc}{alive}}{\getvertex{show}{$d_{13}$}{right=0.04cm}{#1+4.5}{#2-2.4}{z3}}{\getvertex{hide}{$d_{13}$}{right=0.04cm}{#1+4.5}{#2-2.4}{z3}}
	
	\tikzset{EdgeStyle/.style=\za-organism}\Edge[label=$f_8$,labelstyle=\za-label](z0)(z1)
	\tikzset{EdgeStyle/.style=\zb-organism}\Edge[label=$f_9$,labelstyle=\zb-label](z1)(z2)
	\tikzset{EdgeStyle/.style=\zc-organism}\Edge[label=$f_{10}$,labelstyle=\zc-label](z2)(z3)
	
	\ifthenelse{\equal{\za}{alive} \AND \equal{\zb}{alive}}{
		\tikzset{EdgeStyle/.style=composite-organism}\Edge[style={bend right=50}, labelstyle=below](z0)(z2)
		\ifthenelse{\equal{\zc}{alive}}{
			\tikzset{EdgeStyle/.style=composite-organism}\Edge[style={bend right=50}, labelstyle=below](z0)(z3)
		}{}
	}{}
	
	\ifthenelse{\equal{\zb}{alive} \AND \equal{\zc}{alive}}{
		\tikzset{EdgeStyle/.style=composite-organism}\Edge[style={bend left=50}, labelstyle=below](z1)(z3)		
	}{}
		
	\draw (#1-0.5,#2-3.2) rectangle (#1+6.5,#2+1.2);
}

\bstctlcite{IEEEexample:BSTcontrol}

\title{Composition Machines: Programming Self-Organising Software Models for the Emergence of Sequential Program Spaces}

\author{\IEEEauthorblockN{Damian Arellanes}
\IEEEauthorblockA{School of Computing and Communications \\
Lancaster University\\
Lancaster LA1 4WA, United Kingdom \\
\{damian.arellanes@lancaster.ac.uk\}}
}

\maketitle

\begin{abstract}
We are entering a new era in which software systems are becoming more and more complex and larger. So, the composition of such systems is becoming infeasible by manual means. To address this challenge, self-organising software models represent a promising direction since they allow the (bottom-up) emergence of complex computational structures from simple rules. In this paper, we propose an abstract machine, called the composition machine, which allows the definition and the execution of such models. Unlike typical abstract machines, our proposal does not compute individual programs but enables the emergence of multiple programs at once. We particularly present the machine's semantics and provide examples to demonstrate its operation with well-known rules from the realm of Boolean logic and elementary cellular automata. 
\end{abstract}

\begin{IEEEkeywords}
Self-Composition, Self-Organising Software Models, Self-Organisation, Emergent Software, Applied Category Theory
\end{IEEEkeywords}

\section{Introduction}
\label{sec:Introduction}

Building software that builds software is a long-standing challenge which has been considered the holy grail of Computer Science since the inception of Artificial Intelligence~\cite{gulwani_program_2017}. As software systems are becoming increasingly large and complex (e.g., Cyber-Physical Systems~\cite{arellanes_evaluating_2020}), this challenge needs to be seriously considered for avoiding the infeasibility of manual composition of complex software at scale. 

As a first step towards this challenge, we have introduced the notion of \emph{composition by self-organisation} in which complex software models are not explicitly programmed, but emerge from simple rules in a decentralised manner~\cite{arellanes_self-organizing_2021}. In this paper, we propose a discrete-time abstract machine that fits into such a paradigm, called the \emph{composition machine}, for the definition and the execution of \emph{self-organising software models}. In this context, a software model is not a single program performing a specific functionality, but a space of programs evolving through well-defined, deterministic self-organisation rules. Thus, the goal of the composition machine is not the computation of a fixed functionality, but the time-dependent emergence of multiple programs at once. Multiple programs are encapsulated in spaces which are self-similar computational structures rooted in Category Theory. We adopt this formalism for leveraging the large corpus of existing theorems in categorical settings, particularly those related to composition properties.

The rest of the paper is structured as follows. Section~\ref{sec:preliminaries} introduces preliminaries for the formal definition of the proposed machine. Section~\ref{sec:program-spaces} presents the semantics of composition and program spaces. Section~\ref{sec:composition-machines} describes the semantics of the composition machine. Section~\ref{sec:examples} presents examples in which complex program spaces emerge from simple rules. Finally, Section~\ref{sec:conclusions} presents the future directions and the final remarks.

\section{Preliminaries} \label{sec:preliminaries}

This section presents preliminaries about \emph{Category Theory} and \emph{Quiver Theory} for the formal specification of the proposed machine. Our intention is not to provide a deeper analysis, but just to present the most important definitions for our purposes. For a deeper treatment of categorical foundations and quivers, we refer the reader to~\cite{pierce_basic_1991,awodey_category_2010,derksen_introduction_2017}. 

\subsection{Category Theory}

A category consists of a collection of objects, a collection of morphisms between objects, an identity morphism for each object and a way to compose morphisms. There are two laws that composition must satisfy: \emph{associativity} and \emph{unity}. Associativity yields the same composite no matter how composition operands are grouped, and unity leaves unchanged a non-identity morphism when it is composed with an identity morphism. To better understand the semantics of a category, let us describe an example:

\begin{definition}[Category of Sets]
$\mathscr{S}$ is the category where:
\begin{itemize}
\item $\mathscr{S}_0$ is a collection of sets, called objects,
\item $\mathscr{S}_1$ is a collection of functions between sets, called morphisms,
\item $dom: \mathscr{S}_1 \rightarrow \mathscr{S}_0$ is a collection morphism that takes each function $f \in \mathscr{S}_1$ to its domain $dom(f) \in \mathscr{S}_0$,
\item $cod: \mathscr{S}_1 \rightarrow \mathscr{S}_0$ is a collection morphism that takes each function $f \in \mathscr{S}_1$ to its codomain $cod(f) \in \mathscr{S}_0$,
\item there is an identity function $1_X:X \rightarrow X$ for each set $X \in \mathscr{S}_0$ such that $1_X \in \mathscr{S}_1$,
\item there is a composite function $g \circ f \in \mathscr{S}_1$ for every pair $(f,g) \in \mathscr{S}_1 \times \mathscr{S}_1$ satisfying $cod(f)=dom(g)$,
\item \emph{composition is associative}: $\forall (f,g,h) \in \mathscr{S}_1 \times \mathscr{S}_1 \times \mathscr{S}_1, (h \circ g) \circ f = h \circ (g \circ f) \iff cod(f)=dom(g)~\land~cod(g)=dom(h)$, and
\item \emph{composition satisfies unity}: $\forall (f:X \rightarrow Y) \in \mathscr{S}_1, 1_Y \circ f = f = f \circ 1_X$. 
\end{itemize}
\end{definition}

\begin{remark}
We say that a category is \emph{small} if and only if it is internal in $\mathscr{S}$, i.e., the collection of objects and the collection of morphisms are sets. 
\end{remark}

To avoid set-theoretic issues related to the Russell-Zermelo paradox~\cite{russell_principles_1903}, Category Theory defines a category of categories, $Cat$, where objects are small categories and morphisms are functors. A functor is a structure-preserving mapping between two categories, which assigns the objects of one category to the objects of another one. It also assigns the morphisms of one category to the morphisms of the other. Functors preserve both identity morphisms and composition of morphisms.  

\begin{definition}[Functor]
A functor $F:\mathscr{D} \rightarrow \mathscr{E}$ assigns each object $a \in \mathscr{D}_0$ to an object $F(a) \in \mathscr{E}_0$ and each morphism $f \in \mathscr{D}_1$ to a morphism $F(f) \in \mathscr{E}_1$ such that: 
\begin{itemize}
\item $F(1_a)=1_{F(a)}$ for each object $a \in \mathscr{D}_0$, and
\item $F(g \circ f)=F(g) \circ F(f)$ for all composable morphisms $f,g \in \mathscr{D}_1$.
\end{itemize}
\end{definition}

\begin{remark}
A functor $F:\mathscr{D} \rightarrow \mathscr{E}$ is a presentation of a category $\mathscr{D}$ in a category $\mathscr{E}$.
\end{remark}

\begin{definition} [Natural Transformation]
Let $F,G:\mathscr{D} \rightarrow \mathscr{E}$ be two functors each from the category $\mathscr{D}$ to the category $\mathscr{E}$. A natural transformation $\eta:F \rightarrow G$ is a morphism between $F$ and $G$, which satisfies the following:
\begin{itemize}
\item for every object $a \in \mathscr{D}_0$, there is a morphism $\eta_a: F(a) \rightarrow G(a)$ between objects of $\mathscr{E}$, and
\item for every morphism $(f: a \rightarrow b) \in \mathscr{D}_1$, we have $\eta_b \circ F(f)=G(f) \circ \eta_a$.
\end{itemize}
\end{definition}

\subsection{Quiver Foundations}

The underlying structure of a category is typically represented as a \emph{quiver} which is a directed multigraph with loops allowed. The term quiver is used among category theorists to avoid confusions derived from the multiple meanings of the word \emph{graph}. The formal definition of a quiver is presented below, along with other related formalisms.\footnote{Whenever convenient, we treat a function as a set of relations.}

\begin{notation}[Walking Quiver Category]
Let $\mathscr{Q}$ be the walking quiver category consisting of a collection $\mathscr{Q}_0$ of vertices, a collection $\mathscr{Q}_1$ of arrows, a source morphism $s:\mathscr{Q}_1 \rightarrow \mathscr{Q}_0$ and a target morphism $\tau:\mathscr{Q}_1 \rightarrow \mathscr{Q}_0$. 
\end{notation}

\begin{definition}[Quiver]
A \emph{quiver} $Q$ is a functor $\mathscr{Q} \rightarrow \mathscr{S}$. Intuitively, it is a quadruple $(Q_0,Q_1,s,\tau)$ where $Q_0$ is a set of vertices, $Q_1$ is a set of arrows, $s:Q_1 \rightarrow Q_0$ is a source function and $\tau:Q_1 \rightarrow Q_0$ is a target function. If an arrow $\alpha \in Q_1$ has source vertex $x \in Q_0$ and target vertex $y \in Q_0$, we say that $\alpha$ is directed from $x$ to $y$. This is denoted by $\alpha:x \rightarrow y$ with $x=s(\alpha)$ and $y=\tau(\alpha)$.
We call $Q$ finite if and only if the sets $Q_0$ and $Q_1$ are both finite.
\end{definition}

\begin{definition}[Subquiver]
We say that $Q'=(Q'_0,Q'_1,s',\tau')$ is a \emph{subquiver} of $Q=(Q_0,Q_1,s,\tau)$, written $Q' \subseteq Q$, if and only if $Q'_0 \subseteq Q_0$, $Q'_1 \subseteq Q_1$, $s'=s|_{Q'_1}$ and $\tau'=\tau|_{Q'_1}$. 
\end{definition}

\begin{definition}[Non-Trivial Path]
A \emph{non-trivial path} $\rho=(\alpha_n,\alpha_{n-1},\ldots,\alpha_1)$ in a quiver $Q$ is a finite sequence of arrows such that $n \geq 1$ is the length of $\rho$, $\alpha_i \in Q_1$ for all $i \in \mathbb{N} \cap [1,n]$ and $\tau(\alpha_j)=s(\alpha_{j+1})$ for all $j \in \mathbb{N} \cap [1,n-1]$. The source vertex of $\rho$ is denoted by $s(\rho)$ and the target vertex by $\tau(\rho)$ so that $s(\rho)=s(\alpha_1)$ and $\tau(\rho)=\tau(\alpha_n)$ (clearly an abuse of notation but convenient). By convention, all paths are read from right to left as in function composition, and the same arrow can occur more than once in $\rho$.
\end{definition}

\begin{definition}[Trivial Path]
A \emph{trivial path} $\rho_x$ in a quiver $Q$ is just a vertex $x \in Q_0$. Its length is $0$ and there are as many trivial paths as there are vertices in $Q_0$.
\end{definition}

\begin{definition}[Path Concatenation]
Let $\rho_1$ and $\rho_2$ be two paths. We say that $\rho_2 \ast \rho_1$ is a new path, called a path concatenation, if and only if $\tau(\rho_1)=s(\rho_2)$. Unambiguously, the source of $\rho_2 \ast \rho_1$ is denoted by $s(\rho_2 \ast \rho_1)$ and its target by $\tau(\rho_2 \ast \rho_1)$ such that $s(\rho_2 \ast \rho_1)=s(\rho_1)$ and $\tau(\rho_2 \ast \rho_1)=\tau(\rho_2)$.
\end{definition}

\begin{definition}[Set of Paths]
By convention, $Q_m$ is the set of paths of length $m \geq 0$ in a quiver $Q$, and $Q_*=\bigcup\limits_{m \in \mathbb{N}} Q_{m}$ is the set of all the possible paths in $Q$.
\end{definition}

\begin{theorem}[\cite{derksen_introduction_2017}] \label{theorem:infinite-paths}
Given a quiver $Q$, $|Q_*|$ is finite if and only if $Q$ is finite and has no oriented cycles.
\end{theorem}

\begin{definition}[Path Category]
The \emph{path category} $P(Q)$ on a quiver $Q$ consists of:
\begin{itemize}
\item a collection $P(Q)_0$ whose objects are the vertices in $Q_0$,
\item a collection $P(Q)_1$ whose morphisms are the paths in $Q_*$,
\item a collection morphism $s:P(Q)_1 \rightarrow P(Q)_0$ that maps each path $\rho \in P(Q)_1$ to its source vertex $s(\rho) \in P(Q)_0$,
\item a collection morphism $\tau:P(Q)_1 \rightarrow P(Q)_0$ that maps each path $\rho \in P(Q)_1$ to its target vertex $\tau(\rho) \in P(Q)_0$,
\item a trivial path $1_x \in P(Q)_1$ for every vertex $x \in P(Q)_0$, called the identity of $x$, and
\item a path concatenation $\rho_2 \ast \rho_1 \in P(Q)_1$ for every pair $(\rho_1,\rho_2) \in P(Q)_1 \times P(Q)_1$ satisfying $\tau(\rho_1)=s(\rho_2)$.
\end{itemize}
\end{definition}

\begin{remark}
The category $P(Q)$ is sound because: 
\begin{itemize}
\item \emph{composition is associative}: $\forall (\rho_1,\rho_2,\rho_3) \in P(Q)_1 \times P(Q)_1 \times P(Q)_1, (\rho_3 \ast \rho_2) \ast \rho_1 = \rho_3 \ast (\rho_2 \ast \rho_1) \iff \tau(\rho_1)=s(\rho_2)~\land~\tau(\rho_2)=s(\rho_3)$, and 
\item \emph{composition satisfies unity}: $\forall \rho \in P(Q)_1, 1_{\tau(\rho)} \ast \rho = \rho = \rho \ast 1_{s(\rho)}$.
\end{itemize}
\end{remark}

\begin{remark}
Given a path category $P(Q)$, the composition of two morphisms in $P(Q)_1$ is defined by the concatenation of two paths in $Q_*$, and an identity path $1_x \in P(Q)_1$ corresponds to a trivial path $\rho_x \in Q_0$. Also, every path $\rho \in Q_1$ is mapped to its one morphism in $P(Q)_1$.
\end{remark}

\begin{theorem}[\cite{awodey_category_2010}]\label{theorem:adjoint-functor}
Let $\mathscr{S}^{\mathscr{Q}}$ be the category where objects are functors $Q:\mathscr{Q} \rightarrow \mathscr{S}$ (i.e., quivers) and morphisms are natural transformations between such functors. Then, there is a functor $P: \mathscr{S}^{\mathscr{Q}} \rightarrow Cat$ that sends each quiver $Q$ to its corresponding path category $P(Q)$, by taking each vertex $x \in Q_0$ to an object $P(x) \in P(Q)_0$ and each path $\rho \in Q_*$ to a morphism $P(\rho) \in P(Q)_1$.
\end{theorem}

\section{Semantics of Composition and Program Spaces} \label{sec:program-spaces}

In this section, we present the semantics of composition and program spaces for the rest of the paper. For this, we introduce the notion of a \emph{computon} which is the fundamental unit of computation in our proposal.\footnote{The word \emph{computon} derives from the Latin root for computation (\emph{computus}) and the Greek suffix \emph{-on}. In Physics, this suffix is traditionally used to designate subatomic particle names.} Informally, a computon is a sequential program that produces exactly one output value for a given input value via some finite computation. A value is an element of some set referred to as data type.\footnote{Examples of data types include the set of integers, the set of real numbers and the set of characters. A detailed overview of data types is out of the scope of this paper. This is because the definition of type systems depends on the application domain.} Formally: 

\begin{definition} [Data Type]
A data type $d$ is a set of values equipped with at least one \emph{k-ary} operation $d^k \rightarrow d$ that can be performed on those values. 
\end{definition}

\begin{notation}
Let $\mathbb{D}$ be the universal set of data types.
\end{notation}

\begin{definition} [Computon]
A \emph{computon} $f:d_1 \rightarrow d_2$ is a function from an input data type $d_1 \in \mathbb{D}$ to an output data type $d_2 \in \mathbb{D}$. We say that it is an identity computon if $d_1=d_2$ and $\forall v \in d_1, f(v)=v$. Otherwise, $f$ is a non-identity computon. 
\end{definition}

\begin{definition} [Composite Computon] \label{def:composite-computon}
A \emph{composite computon} $f_2 \circ f_1$ is a higher-order function given by the composition of computons $f_1$ and $f_2$ where $cod(f_1)=dom(f_2)$.
\end{definition}

\begin{definition}[Computon Category] \label{def:category-computon}
The \emph{computon category} $\mathscr{C}$ consists of:
\begin{itemize}
\item a collection $\mathscr{C}_0$ of data types,
\item a collection $\mathscr{C}_1$ of computons,
\item a morphism $in:\mathscr{C}_1 \rightarrow \mathscr{C}_0$ that takes each computon $f \in \mathscr{C}_1$ to an input data type $d \in \mathscr{C}_0$,
\item a morphism $out:\mathscr{C}_1 \rightarrow \mathscr{C}_0$ that takes each computon $f \in \mathscr{C}_1$ to an output data type $d \in \mathscr{C}_0$,
\item a composite computon $f_2 \circ f_1 \in \mathscr{C}_1$ for every pair $(f_1,f_2) \in \mathscr{C}_1 \times \mathscr{C}_1$ satisfying $out(f_1)=in(f_2)$, and
\item an identity computon $1_d \in \mathscr{C}_1$ for every data type $d \in \mathscr{C}_0$.
\end{itemize}
\end{definition}

\begin{remark}
The category $\mathscr{C}$ is sound because: 
\begin{itemize}
\item \emph{composition of computons is associative}: $\forall (f_3,f_2,f_1) \in \mathscr{C}_1 \times \mathscr{C}_1 \times \mathscr{C}_1, (f_3 \circ f_2) \circ f_1 = f_3 \circ (f_2 \circ f_1) \iff out(f_1)=in(f_2)~\land~out(f_2)=in(f_3)$, and
\item \emph{composition of computons satisfies unity}: $\forall (f:d_1 \rightarrow d_2) \in \mathscr{C}_1$, $1_{d_2} \circ f = f = f \circ 1_{d_1}$.
\end{itemize}
\end{remark}

\begin{definition} [Program Space] \label{def:program-space}
A \emph{program space} is a functor $R: \mathscr{C} \rightarrow \mathscr{S}$ from the computon category to the set category. It takes the collection $\mathscr{C}_0$ to a set $R(\mathscr{C}_0)$ of data types, the collection $\mathscr{C}_1$ to a set $R(\mathscr{C}_1)$ of computons, the morphism $in: \mathscr{C}_1 \rightarrow \mathscr{C}_0$ to a function $R(in): R(\mathscr{C}_1) \rightarrow R(\mathscr{C}_0)$ and the morphism $out: \mathscr{C}_1 \rightarrow \mathscr{C}_0$ to a function $R(out): R(\mathscr{C}_1) \rightarrow R(\mathscr{C}_0)$. 
\end{definition}

Intuitively, a program space $R$ is a quadruple $(R_0,R_1,in,out)$ where $R_0$ is a set of data types, $R_1$ is a set of computons, $in:R_1 \rightarrow R_0$ is an input function and $out:R_1 \rightarrow R_0$ is an output function. 

\begin{definition} [Transformer] \label{def:transformer}
A \emph{transformer} $T: P(Q) \rightarrow \mathscr{C}$ is a functor for presenting the path category on a quiver $Q$ in a computon category $\mathscr{C}$. It takes the collection $P(Q)_0$ of vertices to a collection $\mathscr{C}_0$ of data types, the collection $P(Q)_1$ of paths to a collection $\mathscr{C}_1$ of computons, the source morphism $s$ to an input morphism $in$ and the target morphism $\tau$ to an output morphism $out$ such that: 
\begin{itemize}
\item there is an identity computon $1_{T(x)} \in \mathscr{C}_1$ for each vertex $x \in P(Q)_0$ where $T(x) \in \mathscr{C}_0$, and
\item there is a composite computon $T(\rho_n \ast \cdots \ast \rho_1) \in \mathscr{C}_1$ for every path concatenation $\rho_n \ast \ldots \ast \rho_1 \in P(Q)_1$ such that $T(\rho_n \ast \cdots \ast \rho_1) = T(\rho_n) \circ \cdots \circ T(\rho_1)$ and $T(\rho_1),\ldots,T(\rho_n) \in \mathscr{C}_1$.
\end{itemize}
\end{definition}

\begin{notation}
$(T \circ P)(Q)$ denotes the computon category generated from the path category on a quiver $Q$. Its collection of data types is denoted by $(T \circ P)(Q)_0$ and its collection of computons by $(T \circ P)(Q)_1$.
\end{notation}

\begin{definition} \label{def:composite-functor-space}
From Definitions~\ref{def:program-space} and \ref{def:transformer}, we can specify a composite functor $R \circ T: P(Q) \rightarrow \mathscr{S}$ for transforming the path category on a quiver $Q$ into its corresponding program space $(R \circ T \circ P)(Q)$. For short, we refer $(R \circ T \circ P)(Q)$ to as the program space on $Q$. 
\end{definition}

\begin{example} \label{example:composition-semantics}
This example is a walkthrough for the construction of the program space $(R \circ T \circ P)(Q)$ on the quiver $Q=(Q_0,Q_1,s,\tau)$. This quiver is illustrated in Figure~\ref{fig:example-quiver} and is defined as follows:
\begin{itemize}
\item $Q_0=\{x_1, x_2, \ldots, x_8\}$,
\item $Q_1=\{\alpha_1, \alpha_2, \ldots, \alpha_6\}$,
\item $s=\{(\alpha_1,x_1),(\alpha_2,x_2),(\alpha_3,x_3),(\alpha_4,x_1),(\alpha_5,x_6),\\(\alpha_6,x_7)\}$, and
\item $\tau=\{(\alpha_1,x_2),(\alpha_2,x_3),(\alpha_3,x_4),(\alpha_4,x_5),(\alpha_5,x_7),\\(\alpha_6,x_8)\}$.
\end{itemize}

\begin{figure}[!h]
\centering
\begin{tikzpicture}
	\GraphInit[vstyle=Classic]
	\tikzset{VertexStyle/.append style=quiver-vertex}
	
	\getvertex{show}{$x_1$}{above=.04cm}{0}{4}{d1}
	\getvertex{show}{$x_2$}{above=.04cm}{2}{4}{d2}
	\getvertex{show}{$x_3$}{above=.04cm}{4}{4}{d3}
	\getvertex{show}{$x_4$}{right=.04cm}{6}{4}{d4}
	\getvertex{show}{$x_5$}{left=.04cm}{-1.5}{4}{d5}
	\getvertex{show}{$x_6$}{left=.04cm}{0}{2.3}{d6}
	\getvertex{show}{$x_7$}{below=.04cm}{2}{2.3}{d7}
	\getvertex{show}{$x_8$}{right=.04cm}{4}{2.3}{d8}
	
	\tikzset{EdgeStyle/.style=alive-organism}\Edge[label=$\alpha_1$,labelstyle=alive-label](d1)(d2)
	\tikzset{EdgeStyle/.style=alive-organism}\Edge[label=$\alpha_2$,labelstyle=alive-label](d2)(d3)
	\tikzset{EdgeStyle/.style=alive-organism}\Edge[label=$\alpha_3$,labelstyle=alive-label](d3)(d4)
	\tikzset{EdgeStyle/.style=alive-organism}\Edge[label=$\alpha_4$,labelstyle=alive-label](d1)(d5)
	\tikzset{EdgeStyle/.style=alive-organism}\Edge[label=$\alpha_5$,labelstyle=alive-label](d6)(d7)
	\tikzset{EdgeStyle/.style=alive-organism}\Edge[label=$\alpha_6$,labelstyle=alive-label](d7)(d8)
\end{tikzpicture}
\caption{Graphical representation of the quiver $Q$ described in Example~\ref{example:composition-semantics}.}
\label{fig:example-quiver}
\end{figure}
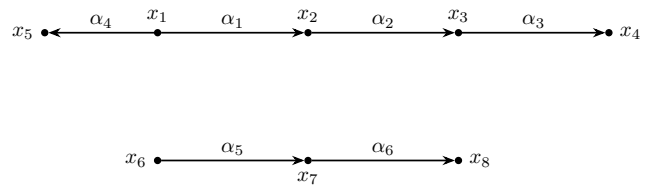

Once the quiver $Q$ has been defined, we can transform it into a computon category as follows. Let $D=\{d_1, d_2, \ldots, d_8\}$ be a set of data types and $F=\{f_1: d_1 \rightarrow d_2, f_2: d_2 \rightarrow d_3, f_3: d_3 \rightarrow d_4, f_4: d_1 \rightarrow d_5, f_5: d_6 \rightarrow d_7, f_6: d_7 \rightarrow d_8\}$ be a set of computons. According to Definition~\ref{def:transformer}, we can specify a transformer $T$ to present $P(Q)$ in a computon category $(T \circ P)(Q)$. We accomplish this by mapping each vertex $P(x_i) \in P(Q)_0$ to a data type $d_i \in D$ for all $i \in \mathbb{N} \cap [1,8]$ and each path $P(\alpha_j) \in P(Q)_1$ of length one to a computon $f_j \in F$ for all $j \in \mathbb{N} \cap [1,6]$. Composite computons correspond to path concatenations, and identity computons are trivial paths. This construction is formally defined as follows:
\begin{itemize}
\item $\forall i \in \mathbb{N} \cap [1,8], x_i \in Q_0 \implies T(P(x_i)) \in (T \circ P)(Q)_0$ such that $P(x_i) \in P(Q)_0$ and $T(P(x_i)) = d_i \in D$,
\item $\forall i \in \mathbb{N} \cap [1,8], x_i \in Q_0 \implies T(1_{P(x_i)}):T(P(x_i)) \rightarrow T(P(x_i))$ such that $1_{P(x_i)} \in P(Q)_1, T(1_{P(x_i)}) \in (T \circ P)(Q)_1, P(x_i) \in P(Q)_0, T(P(x_i)) \in (T \circ P)(Q)_0$ and $T(P(x_i)) = d_i \in D$,
\item $\forall i \in \mathbb{N} \cap [1,6], \alpha_i \in Q_1 \implies T(P(\alpha_i)) \in (T \circ P)(Q)_1$ such that $P(\alpha_i) \in P(Q)_1$ and $T(P(\alpha_i)) = f_i \in F$, and
\item $\forall (\rho_n \ast \cdots \ast \rho_1) \in P(Q)_1, [T(\rho_n) \circ \cdots \circ T(\rho_1)] \in (T \circ P)(Q)_1$.
\end{itemize}

Finally, the resulting computon category $(T \circ P)(Q)$ can be presented in $\mathscr{S}$ via the functor $R$ (see Definition~\ref{def:program-space}). Intuitively, this is the program space $(R \circ T \circ P)(Q)=((R \circ T \circ P)(Q)_0,(R \circ T \circ P)(Q)_1,in,out)$ where:
\begin{itemize}
\item $(R \circ T \circ P)(Q)_0=D$,
\item $(R \circ T \circ P)(Q)_1=F \cup \{f_2 \circ f_1: d_1 \rightarrow d_3, f_3 \circ f_2: d_2 \rightarrow d_4, f_3 \circ f_2 \circ f_1: d_1 \rightarrow d_4, f_6 \circ f_5: d_6 \rightarrow d_8, 1_{d_1}: d_1 \rightarrow d_1, 1_{d_2}: d_2 \rightarrow d_2, \ldots, 1_{d_8}: d_8 \rightarrow d_8\}$,
\item $in=\{(f_1,d_1),(f_2,d_2),(f_3,d_3),(f_4,d_1),(f_5,d_6),\\(f_6,d_7),(f_2 \circ f_1,d_1),(f_3 \circ f_2,d_2),(f_3 \circ f_2 \circ f_1,d_1),\\(f_6 \circ f_5,d_6),(1_{d_1}, d_1),\ldots,(1_{d_8}, d_8)\}$, and
\item $out=\{(f_1,d_2),(f_2,d_3),(f_3,d_4),(f_4,d_5),(f_5,d_7),\\(f_6,d_8),(f_2 \circ f_1,d_3),(f_3 \circ f_2,d_4),(f_3 \circ f_2 \circ f_1,d_4),\\(f_6 \circ f_5,d_8),(1_{d_1}, d_1),\ldots,(1_{d_8}, d_8)\}$.
\end{itemize}

The program space $(R \circ T \circ P)(Q)$ is illustrated in Figure~\ref{fig:example-computon-space}. Although we omit identity computons for clarity, the above definition reveals that the identities are elements of the set $(R \circ T \circ P)(Q)_1$. Without loss of generality, hereafter we do not show identity computons when depicting program spaces.

\begin{figure}[!h]
\centering
\begin{tikzpicture}
	\GraphInit[vstyle=Classic]
	\tikzset{VertexStyle/.append style=category-vertex}
	
	\getvertex{show}{$d_1$}{below=.04cm}{0}{4}{d1}
	\getvertex{show}{$d_2$}{below=.04cm}{2}{4}{d2}
	\getvertex{show}{$d_3$}{below=.04cm}{4}{4}{d3}
	\getvertex{show}{$d_4$}{right=.04cm}{6}{4}{d4}
	\getvertex{show}{$d_5$}{left=.04cm}{-1.5}{4}{d5}
	\getvertex{show}{$d_6$}{left=.04cm}{0}{2.3}{d6}
	\getvertex{show}{$d_7$}{below=.04cm}{2}{2.3}{d7}
	\getvertex{show}{$d_8$}{right=.04cm}{4}{2.3}{d8}
	
	\tikzset{EdgeStyle/.style=alive-organism}\Edge[label=$f_1$,labelstyle=alive-label](d1)(d2)
	\tikzset{EdgeStyle/.style=alive-organism}\Edge[label=$f_2$,labelstyle=alive-label](d2)(d3)
	\tikzset{EdgeStyle/.style=alive-organism}\Edge[label=$f_3$,labelstyle=alive-label](d3)(d4)
	\tikzset{EdgeStyle/.style=alive-organism}\Edge[label=$f_4$,labelstyle=alive-label](d1)(d5)
	\tikzset{EdgeStyle/.style=alive-organism}\Edge[label=$f_5$,labelstyle=alive-label](d6)(d7)
	\tikzset{EdgeStyle/.style=alive-organism}\Edge[label=$f_6$,labelstyle=alive-label](d7)(d8)
	
	\tikzset{EdgeStyle/.style=composite-organism}\Edge[style={bend left=30}, label=$f_2 \circ f_1$, labelstyle=alive-label](d1)(d3)
	\tikzset{EdgeStyle/.style=composite-organism}\Edge[style={bend left=50}, label=$f_3 \circ f_2 \circ f_1$, labelstyle=alive-label](d1)(d4)
	\tikzset{EdgeStyle/.style=composite-organism}\Edge[style={bend right=30}, label=$f_3 \circ f_2$, labelstyle={below=-1,fill=none,scale=0.8}](d2)(d4)
	\tikzset{EdgeStyle/.style=composite-organism}\Edge[style={bend left=30}, label=$f_6 \circ f_5$, labelstyle={above=-1,fill=none,scale=0.8}](d6)(d8)
\end{tikzpicture}
\caption{Graphical representation of the program space $(R \circ T \circ P)(Q)$ described in Example~\ref{example:composition-semantics}.}
\label{fig:example-computon-space}
\end{figure}
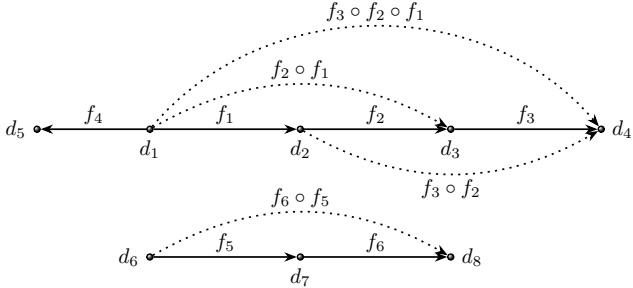
\end{example}

Example~\ref{example:composition-semantics} just shows a space in which the set of data types and the set of computons are finite. Nevertheless, it is possible to construct spaces of finite data types and infinite computons. For this, it suffices to define oriented cycles in the interdependence of computons (see Theorem~\ref{theorem:infinite-paths}). For instance, we can transform the finite space $(R \circ T \circ P)(Q)$ into an infinite one by defining the computon $f_{7}:d_4 \rightarrow d_1$ in $F$. 

\section{Composition Machines}
\label{sec:composition-machines}

We propose the notion of composition machines for the emergence of sequential program spaces. These machines allow the definition and the execution of self-organising software models~\cite{arellanes_self-organizing_2021}, and represent a shift from the computation of single programs to the emergence of program spaces.

\begin{definition} \label{def:machine} [Composition Machine]
A \emph{composition machine} $M$ is a septuple $(D,F,Q,\mu,S,N,\delta)$ where:
\begin{itemize}
\item $D$ is a non-empty finite set of data types,
\item $F$ is a non-empty finite set of computons such that $\forall f,g \in F, cod(f) = cod(g) \iff f = g$ and $\forall f \in F, dom(f) \neq cod(f)$,
\item $Q$ is an acyclic quiver where each vertex $x \in Q_0$ represents a data type $d \in D$ and each arrow $\alpha \in Q_1$ (called an organism) represents a computon $f \in F$ such that $\forall \alpha_1,\alpha_2 \in Q_1, \tau(\alpha_1) = \tau(\alpha_2) \iff \alpha_1 = \alpha_2$ and $\forall \alpha \in Q_1, s(\alpha) \neq \tau(\alpha)$,
\item $\mu$ consists of a pair of bijective functions, $\mu_0: Q_0 \rightarrow D$ and $\mu_1: Q_1 \rightarrow F$, for mapping vertices to data types and organisms to computons, respectively,
\item $S=\{0,1\}$ is the set of possible organism states,
\item $N=\{N_1,N_2,N_3,N_4\}$ is a non-empty finite collection of neighbourhoods:
\begin{itemize}
\item $N_1 \subseteq Q_1$ is a set where each element $(\alpha) \in N_1$ is the unary neighbourhood of an isolated organism $\alpha \in Q_1$ that has has no neighbours to the right and no neighbours to the left, i.e., $\forall (\alpha) \in N_1, \nexists \alpha_0 \in Q_1, \tau(\alpha_0)=s(\alpha)~\lor~\tau(\alpha)=s(\alpha_0)$.
\item $N_2 \subseteq Q_1 \times Q_1$ is a set where each element $(\alpha_1,\alpha_2) \in N_2$ is the binary neighbourhood of an organism $\alpha_1 \in Q_1$ that has only one neighbour to the right, i.e., $\forall (\alpha_1,\alpha_2) \in N_2, (\nexists \alpha_0 \in Q_1, \tau(\alpha_0)=s(\alpha_1))~\land~\tau(\alpha_1)=s(\alpha_2)$.
\item $N_3 \subseteq Q_1 \times Q_1$ is a set where each element $(\alpha_1,\alpha_2) \in N_3$ is the binary neighbourhood of an organism $\alpha_2 \in Q_1$ that has only one neighbour to the left, i.e., $\forall (\alpha_1,\alpha_2) \in N_3, (\nexists \alpha_0 \in Q_1, \tau(\alpha_2)=s(\alpha_0))~\land~\tau(\alpha_1)=s(\alpha_2)$.
\item $N_4 \subseteq Q_1 \times Q_1 \times Q_1$ is a set where each element $(\alpha_1,\alpha_2, \alpha_3) \in N_4$ is the ternary neighbourhood of an organism $\alpha_2 \in Q_1$ which has exactly one neighbour to the right and exactly one neighbour to the left, i.e., $\forall (\alpha_1,\alpha_2,\alpha_3) \in N_4, \tau(\alpha_1)=s(\alpha_2)~\land~\tau(\alpha_2)=s(\alpha_3)$.
\end{itemize}
\item $\delta$ consists of four local state transition functions:
\begin{itemize}
\item $\delta_1: S \rightarrow S$ for every \emph{isolated organism} $\alpha \in Q_1$ with a neighbourhood $(\alpha) \in N_1$,
\item $\delta_2: S^2 \rightarrow S$ for every organism $\alpha_1 \in Q_1$ with a neighbourhood $(\alpha_1,\alpha_2) \in N_2$,
\item $\delta_3: S^2 \rightarrow S$ for every organism $\alpha_2 \in Q_1$ with a neighbourhood $(\alpha_1,\alpha_2) \in N_3$, and
\item $\delta_4: S^3 \rightarrow S$ for every organism $\alpha_2 \in Q_1$ with a neighbourhood $(\alpha_1,\alpha_2, \alpha_3) \in N_4$.
\end{itemize} 
\end{itemize}
\end{definition}

Intuitively, a composition machine $M=(D,F,Q,\mu,S,N,\delta)$ consists of $|F|$ \emph{organisms} operating in discrete moments in time.\footnote{As the function $\mu_1$ is bijective (see Definition~\ref{def:machine}), we have that $|F|=|Q_1|$.} Each organism represents a computon $f \in F$ and has at most three neighbours (including itself). We say that two different organisms $\alpha_1 \in Q_1$ and $\alpha_2 \in Q_1$ are neighbours if there is a non-trivial path $(\alpha_1,\alpha_2)$ or a non-trivial path $(\alpha_2,\alpha_1)$. Equivalently, two different organisms are neighbours if the computons they represent are composable (see Definition~\ref{def:composite-computon}).

At each time step $t \in \mathbb{N}$, an organism $\alpha \in Q_1$ is in a state $c(\alpha)^t \in S$. If $c(\alpha)^t=1$, we say that $\alpha$ is alive; otherwise, we say that $\alpha$ is dead (see Figure~\ref{fig:organism-states}).

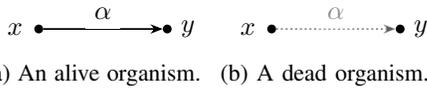
\begin{figure}[!h]
\centering
\subcaptionbox{An alive organism. \label{fig:organism-alive}}
{
  \begin{tikzpicture}
	\GraphInit[vstyle=Classic]
	\tikzset{VertexStyle/.append style=quiver-vertex}

	\Vertex[L=\hbox{$x$}, style={left=.15cm}, x=0.0cm, y=2.5cm]{v0}
	\Vertex[L=\hbox{$y$}, x=1.5cm, y=2.5cm]{v1}

	\tikzset{EdgeStyle/.style=alive-organism}
	\Edge[labelstyle=above, label=$\alpha$](v0)(v1)
  \end{tikzpicture} 
}
\subcaptionbox{A dead organism. \label{fig:organism-dead}}
{
  \begin{tikzpicture}
	\GraphInit[vstyle=Classic]
	\tikzset{VertexStyle/.append style=quiver-vertex}

	\Vertex[L=\hbox{$x$}, style={left=.15cm}, x=2cm, y=2.5cm]{v0}
	\Vertex[L=\hbox{$y$}, x=3.5cm, y=2.5cm]{v1}

	\tikzset{EdgeStyle/.style=dead-organism}
    \Edge[labelstyle=above, label=$\alpha$](v0)(v1)
  \end{tikzpicture} 
}
\caption{At time $t$, an organism $\alpha:x \rightarrow y$ is (a) alive if $c(\alpha)^t=1$ or (b) dead if $c(\alpha)^t=0$.}
\label{fig:organism-states}
\end{figure}

The configuration of $M$ at $t$ is then a function $c:Q_1 \rightarrow S$ that assigns a state to each organism $\alpha \in Q_1$. Intuitively, it is a snapshot of all the states in the system of organisms at some moment in time. Initially, at $t=0$, $M$ is in the so-called \emph{initial configuration} and, in subsequent steps, the states of its organisms are updated in parallel according to four $\delta$ functions. 

A $\delta$ function is a local update rule that defines the next state of an organism based on the current state of $n$ neighbours. More concretely, if the neighbours of an organism have states $c(\alpha_1)^t,c(\alpha_2)^t,\ldots,c(\alpha_n)^t$ at $t$, then the state of that organism at $t+1$ is given by $\delta_i(c(\alpha_1)^t,c(\alpha_2)^t,\ldots,c(\alpha_n)^t)$ such that $i \in \mathbb{N} \cap [1,4]$ and $n \in \mathbb{N} \cap [1,3]$. Particularly, $n=1$ for unary neighbourhoods, $n=2$ for binary neighbourhoods and $n=3$ for ternary neighbourhoods. 

Moreover, as an organism can be alive or dead, there are $2^1$, $2^2$ and $2^3$ possible patterns for unary, binary and ternary neighbourhoods, respectively (see Figure~\ref{fig:neighbourhood-configurations}). So, $|dom(\delta_1)|=|S|=2$, $|dom(\delta_2)|=|dom(\delta_3)|=|S|^2=4$ and $|dom(\delta_4)|=|S|^3=8$. Accordingly, there are $2^2 \times 2^4 \times 2^8 = 2^{14}$ possible rules for a composition machine.

\begin{figure}[!h]
\centering
\subcaptionbox{\label{fig:neighbourhood-configurations-ternary}}[0.4\linewidth]
{
  \begin{tikzpicture}
	\GraphInit[vstyle=Classic]
	\tikzset{VertexStyle/.append style=quiver-vertex}
	
	\Vertex[L=\hbox{}, x=0.0cm, y=4.1cm]{v0}
	\Vertex[L=\hbox{}, x=1.0cm, y=4.1cm]{v1}
	\Vertex[L=\hbox{}, x=2.0cm, y=4.1cm]{v2}
	\Vertex[L=\hbox{}, x=3.0cm, y=4.1cm]{v3}
	\tikzset{EdgeStyle/.style=alive-organism}
	\Edge[labelstyle=below](v0)(v1)
	\tikzset{EdgeStyle/.style=alive-organism}
	\Edge[labelstyle=below](v1)(v2)
	\tikzset{EdgeStyle/.style=alive-organism}
	\Edge[labelstyle=below](v2)(v3)
	\draw (-0.1,3.9) rectangle (3.1,4.3);	
	
	\Vertex[L=\hbox{}, x=0.0cm, y=3.5cm]{v4}
	\Vertex[L=\hbox{}, x=1.0cm, y=3.5cm]{v5}
	\Vertex[L=\hbox{}, x=2.0cm, y=3.5cm]{v6}
	\Vertex[L=\hbox{}, x=3.0cm, y=3.5cm]{v7}
	\tikzset{EdgeStyle/.style=alive-organism}
	\Edge[labelstyle=below](v4)(v5)
	\tikzset{EdgeStyle/.style=alive-organism}
	\Edge[labelstyle=below](v5)(v6)
	\tikzset{EdgeStyle/.style=dead-organism}
	\Edge[labelstyle=below](v6)(v7)
	\draw (-0.1,3.3) rectangle (3.1,3.7);	
	
	\Vertex[L=\hbox{}, x=0.0cm, y=2.9cm]{v8}
	\Vertex[L=\hbox{}, x=1.0cm, y=2.9cm]{v9}
	\Vertex[L=\hbox{}, x=2.0cm, y=2.9cm]{v10}
	\Vertex[L=\hbox{}, x=3.0cm, y=2.9cm]{v11}
	\tikzset{EdgeStyle/.style=alive-organism}
	\Edge[labelstyle=below](v8)(v9)
	\tikzset{EdgeStyle/.style=dead-organism}
	\Edge[labelstyle=below](v9)(v10)
	\tikzset{EdgeStyle/.style=alive-organism}
	\Edge[labelstyle=below](v10)(v11)
	\draw (-0.1,2.7) rectangle (3.1,3.1);
	
	\Vertex[L=\hbox{}, x=0.0cm, y=2.3cm]{v12}
	\Vertex[L=\hbox{}, x=1.0cm, y=2.3cm]{v13}
	\Vertex[L=\hbox{}, x=2.0cm, y=2.3cm]{v14}
	\Vertex[L=\hbox{}, x=3.0cm, y=2.3cm]{v15}
	\tikzset{EdgeStyle/.style=alive-organism}
	\Edge[labelstyle=below](v12)(v13)
	\tikzset{EdgeStyle/.style=dead-organism}
	\Edge[labelstyle=below](v13)(v14)
	\tikzset{EdgeStyle/.style=dead-organism}
	\Edge[labelstyle=below](v14)(v15)
	\draw (-0.1,2.1) rectangle (3.1,2.5);	
	
	\Vertex[L=\hbox{}, x=0.0cm, y=1.7cm]{v16}
	\Vertex[L=\hbox{}, x=1.0cm, y=1.7cm]{v17}
	\Vertex[L=\hbox{}, x=2.0cm, y=1.7cm]{v18}
	\Vertex[L=\hbox{}, x=3.0cm, y=1.7cm]{v19}
	\tikzset{EdgeStyle/.style=dead-organism}
	\Edge[labelstyle=below](v16)(v17)
	\tikzset{EdgeStyle/.style=alive-organism}
	\Edge[labelstyle=below](v17)(v18)
	\tikzset{EdgeStyle/.style=alive-organism}
	\Edge[labelstyle=below](v18)(v19)
	\draw (-0.1,1.5) rectangle (3.1,1.9);
	
	\Vertex[L=\hbox{}, x=0.0cm, y=1.1cm]{v20}
	\Vertex[L=\hbox{}, x=1.0cm, y=1.1cm]{v21}
	\Vertex[L=\hbox{}, x=2.0cm, y=1.1cm]{v22}
	\Vertex[L=\hbox{}, x=3.0cm, y=1.1cm]{v23}
	\tikzset{EdgeStyle/.style=dead-organism}
	\Edge[labelstyle=below](v20)(v21)
	\tikzset{EdgeStyle/.style=alive-organism}
	\Edge[labelstyle=below](v21)(v22)
	\tikzset{EdgeStyle/.style=dead-organism}
	\Edge[labelstyle=below](v22)(v23)
	\draw (-0.1,0.9) rectangle (3.1,1.3);
	
	\Vertex[L=\hbox{}, x=0.0cm, y=0.6cm]{v24}
	\Vertex[L=\hbox{}, x=1.0cm, y=0.6cm]{v25}
	\Vertex[L=\hbox{}, x=2.0cm, y=0.6cm]{v26}
	\Vertex[L=\hbox{}, x=3.0cm, y=0.6cm]{v27}
	\tikzset{EdgeStyle/.style=dead-organism}
	\Edge[labelstyle=below](v24)(v25)
	\tikzset{EdgeStyle/.style=dead-organism}
	\Edge[labelstyle=below](v25)(v26)
	\tikzset{EdgeStyle/.style=alive-organism}
	\Edge[labelstyle=below](v26)(v27)
	\draw (-0.1,0.4) rectangle (3.1,0.8);
	
	\Vertex[L=\hbox{}, x=0.0cm, y=0.0cm]{v28}
	\Vertex[L=\hbox{}, x=1.0cm, y=0.0cm]{v29}
	\Vertex[L=\hbox{}, x=2.0cm, y=0.0cm]{v30}
	\Vertex[L=\hbox{}, x=3.0cm, y=0.0cm]{v31}
	\tikzset{EdgeStyle/.style=dead-organism}
	\Edge[labelstyle=below](v28)(v29)
	\tikzset{EdgeStyle/.style=dead-organism}
	\Edge[labelstyle=below](v29)(v30)
	\tikzset{EdgeStyle/.style=dead-organism}
	\Edge[labelstyle=below](v30)(v31)
	\draw (-0.1,-0.2) rectangle (3.1,0.2);
  \end{tikzpicture} 
}
\subcaptionbox{\label{fig:neighbourhood-configurations-binary}}[0.25\linewidth]
{
  \begin{tikzpicture}
	\GraphInit[vstyle=Classic]
	\tikzset{VertexStyle/.append style=quiver-vertex}

	\Vertex[L=\hbox{}, x=0.0cm, y=4.1cm]{v0}
	\Vertex[L=\hbox{}, x=1.0cm, y=4.1cm]{v1}
	\Vertex[L=\hbox{}, x=2.0cm, y=4.1cm]{v2}
	\tikzset{EdgeStyle/.style=alive-organism}
	\Edge[labelstyle=below](v0)(v1)
	\tikzset{EdgeStyle/.style=alive-organism}
	\Edge[labelstyle=below](v1)(v2)
	\draw (-0.1,3.9) rectangle (2.1,4.3);
	
	\Vertex[L=\hbox{}, x=0.0cm, y=3.5cm]{v3}
	\Vertex[L=\hbox{}, x=1.0cm, y=3.5cm]{v4}
	\Vertex[L=\hbox{}, x=2.0cm, y=3.5cm]{v5}
	\tikzset{EdgeStyle/.style=alive-organism}
	\Edge[labelstyle=below](v3)(v4)
	\tikzset{EdgeStyle/.style=dead-organism}
	\Edge[labelstyle=below](v4)(v5)		
	\draw (-0.1,3.3) rectangle (2.1,3.7);
	
	\Vertex[L=\hbox{}, x=0.0cm, y=2.9cm]{v6}
	\Vertex[L=\hbox{}, x=1.0cm, y=2.9cm]{v7}
	\Vertex[L=\hbox{}, x=2.0cm, y=2.9cm]{v8}
	\tikzset{EdgeStyle/.style=dead-organism}
	\Edge[labelstyle=below](v6)(v7)
	\tikzset{EdgeStyle/.style=alive-organism}
	\Edge[labelstyle=below](v7)(v8)		
	\draw (-0.1,2.7) rectangle (2.1,3.1);
	
	\Vertex[L=\hbox{}, x=0.0cm, y=2.3cm]{v9}
	\Vertex[L=\hbox{}, x=1.0cm, y=2.3cm]{v10}
	\Vertex[L=\hbox{}, x=2.0cm, y=2.3cm]{v11}
	\tikzset{EdgeStyle/.style=dead-organism}
	\Edge[labelstyle=below](v9)(v10)
	\tikzset{EdgeStyle/.style=dead-organism}
	\Edge[labelstyle=below](v10)(v11)		
	\draw (-0.1,2.1) rectangle (2.1,2.5);
  \end{tikzpicture} 
}
\subcaptionbox{\label{fig:neighbourhood-configurations-unary}}[0.3\linewidth]
{
  \begin{tikzpicture}
	\GraphInit[vstyle=Classic]
	\tikzset{VertexStyle/.append style=quiver-vertex}

	\Vertex[L=\hbox{}, x=0.0cm, y=3.5cm]{v0}
	\Vertex[L=\hbox{}, x=1.0cm, y=3.5cm]{v1}
	\tikzset{EdgeStyle/.style=alive-organism}
	\Edge[labelstyle=below](v0)(v1)
	\draw (-0.1,3.3) rectangle (1.1,3.7);
	
	\Vertex[L=\hbox{}, x=0.0cm, y=2.9cm]{v1}
	\Vertex[L=\hbox{}, x=1.0cm, y=2.9cm]{v2}
	\tikzset{EdgeStyle/.style=dead-organism}
	\Edge[labelstyle=below](v1)(v2)
	\draw (-0.1,2.7) rectangle (1.1,3.1);
  \end{tikzpicture} 
}
\caption{Possible patterns for (a) ternary neighbourhoods, (b) binary neighbourhoods and (c) unary neighbourhoods.}
\label{fig:neighbourhood-configurations}
\end{figure}
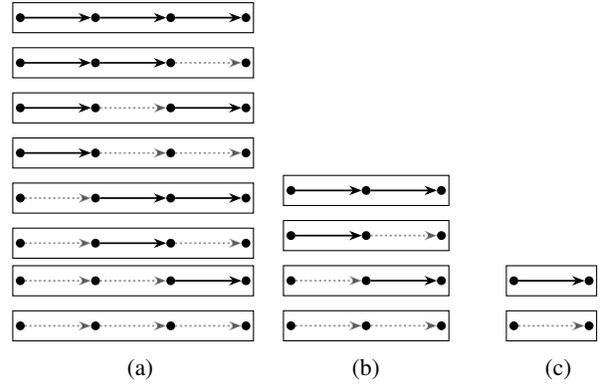

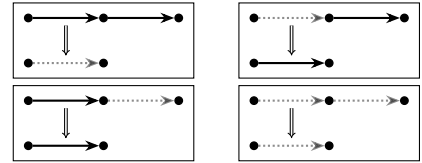
\begin{figure}[!h]
\centering
\begin{tikzpicture}
	\machinerulebinary{0}{4.1}{alive}{alive}{dead}{left}
	\machinerulebinary{0}{3}{alive}{dead}{alive}{left}
	\machinerulebinary{3}{4.1}{dead}{alive}{alive}{left}
	\machinerulebinary{3}{3}{dead}{dead}{dead}{left}		
\end{tikzpicture}
\caption{An example of a local transition rule for an organism that has only one neighbour to the right. This rule is equivalent to a XOR operation and corresponds to the $\delta_2$ function described in Equation~\ref{eq:rule-xor-example}.}
\label{fig:rule-xor-example}
\end{figure}

To explain how rules can be defined, let us consider a simple transition function for organisms that have only one neighbour to the right:

\begin{equation} \label{eq:rule-xor-example}
\delta_2(c(\alpha_1)^t,c(\alpha_2)^t) = c(\alpha_1)^t \oplus c(\alpha_2)^t
\end{equation}
where $\oplus$ is the exclusive-or (XOR) operator defined in Boolean algebra.

Without loss of generality, let us assume that $(\alpha_1,\alpha_2) \in N_2$ is the binary neighbourhood of some organism $\alpha_1$. If we say that both organisms are alive at some moment in time $t$, then the next state of $\alpha_1$ would be given by $c(\alpha_1)^{t+1}=\delta_2(1,1)=0$, as shown in Figure~\ref{fig:rule-xor-example}.

Figure~\ref{fig:rule-xor-example} shows an example of a local transition rule for the binary neighbourhood $N_2$ in some machine $M$. In a similar fashion, we can define rules for the other neighbourhoods to independently transform organisms' states. These transformations occur by simultaneously applying the appropriate rule for each $\alpha \in Q_1$, leading to a time evolution of the global configuration of $M$. Particularly, evolving a configuration $c$ into another one is given by a global transition function $G$:

\begin{equation} \label{eq:global-transition}
G: S^{|Q_1|} \rightarrow S^{|Q_1|}
\end{equation}

So, the time evolution of $M$ is a direct consequence of the repeated application of $G$:

\begin{equation}
c \mapsto G(c) \mapsto (G \circ G)(c) \mapsto (G \circ G \circ G)(c) \mapsto \ldots
\end{equation}

As $c$ is the initial configuration, we can deduce that the machine's orbit $orb(c)$ is equal to:

\begin{equation}
c, G(c), (G \circ G)(c), (G \circ G \circ G)(c), \ldots
\end{equation}

In this case, time refers to the number of applications of $G$. For example, the configuration $c$ exists at time $t=0$, the configuration $G(c)$ is present at time $t=1$, the configuration $(G \circ G)(c)$ appears at time $t=2$, and so on. Indeed, the time evolution of a composition machine is similar to the one exhibited by a synchronous cellular automaton~\cite{wolfram_new_2002}. The difference lies in their underlying semantics. 

In a composition machine $M$, each organism $\alpha \in Q_1$ is causally related to some computon $f \in F$. Furthermore, the global state of $M$ at $t$ is semantically equivalent to a category which defines all the possible sequential compositions at that moment in time. So, the machine's purpose is not the computation of a single program but the emergence of a whole program space. More formally, a program space in $M$ at $t$ is generated from the path category $P(\overrightarrow{Q}^t)$, where $\overrightarrow{Q}^t$ is referred to as the \emph{alive quiver} at $t$.

\begin{definition} [Alive Quiver] \label{def:alive-quiver}
Let $M=(D,F,Q,\mu,S,N,\delta)$ be a composition machine and $c$ be some machine configuration at $t$. The \emph{alive quiver} $\overrightarrow{Q}^t \subseteq Q$ is then a quadruple $(\overrightarrow{Q}_0,\overrightarrow{Q}_1,\overrightarrow{s},\overrightarrow{\tau})$ where:
\begin{itemize}
\item $\overrightarrow{Q}_0=\{x \in Q_0 \mid \exists \alpha \in Q_1, c(\alpha)^t=1~\land~s(\alpha)=x~\lor~\tau(\alpha)=x\}$,
\item $\overrightarrow{Q}_1=\{\alpha \in Q_1 \mid c(\alpha)^t=1\}$,
\item $\overrightarrow{s}=\{(\alpha,s(\alpha)) \in Q_1 \times Q_0 \mid c(\alpha)^t=1 \}$, and
\item $\overrightarrow{\tau}=\{(\alpha,\tau(\alpha)) \in Q_1 \times Q_0 \mid c(\alpha)^t=1 \}$.
\end{itemize}
\end{definition}

The path category on an alive quiver $\overrightarrow{Q}^t$ generates the space $(R \circ T \circ P)(\overrightarrow{Q}^t)$ of all the possible sequential computons at time $t$. As per Definition~\ref{def:composite-functor-space}, this space is produced through a computon category $(T \circ P)(\overrightarrow{Q}^t)$ which we refer to as the \emph{alive category} at $t$.

\begin{definition} [Alive Category] \label{def:alive-category}
Given a composition machine $M=(D,F,Q,\mu,S,N,\delta)$ and some alive quiver $\overrightarrow{Q}^t=(\overrightarrow{Q}_0,\overrightarrow{Q}_1,\overrightarrow{s},\overrightarrow{\tau})$, the \emph{alive category} $(T \circ P)(\overrightarrow{Q}^t)$ generated from $P(\overrightarrow{Q}^t)$ consists of:
\begin{itemize}
\item a collection $(T \circ P)(\overrightarrow{Q}^t)_0$ of data types such that $\forall x \in P(\overrightarrow{Q}^t)_0, \exists ! d \in (T \circ P)(\overrightarrow{Q}^t)_0, d = \mu_0(x)$,
\item a collection $(T \circ P)(\overrightarrow{Q}^t)_1$ of computons where:
\begin{itemize}
\item $\forall x \in P(\overrightarrow{Q}^t)_0, \exists ! 1_d \in (T \circ P)(\overrightarrow{Q}^t)_1, d = \mu_0(x)$, and
\item $\forall (\alpha_n,\ldots,\alpha_1) \in P(\overrightarrow{Q}^t)_1, \exists ! (f_n \circ \cdots \circ f_1) \in (T \circ P)(\overrightarrow{Q}^t)_1, f_1 = \mu_1(\alpha_1)~\land~\cdots~\land~f_n = \mu_1(\alpha_n)$ with $n \geq 1$.
\end{itemize} 
\end{itemize}
\end{definition}

\begin{theorem} \label{theorem:maximal-space}
Given a composition machine $M=(D,F,Q,\mu,S,N,\delta)$, $(R \circ T \circ P)(Q)$ is the maximal program space defining all the possible sequential programs that can exist in $M$ at any moment in time. 
\end{theorem}
\begin{proof}
Let $M=(D,F,Q,\mu,S,N,\delta)$ be a composition machine and $\overrightarrow{Q}^t$ be some alive quiver at time $t$. By Definition~\ref{def:alive-quiver}, we know that $\overrightarrow{Q}^t \subseteq Q$. Then, it trivially follows that $\overrightarrow{Q}^t=Q$ is the largest subset and therefore the maximal alive quiver that can be formed at any moment in time. By Definition~\ref{def:composite-functor-space}, it follows that we can construct the maximal program space $(R \circ T \circ P)(Q)$ as required.
\end{proof}

\section{Examples}
\label{sec:examples}

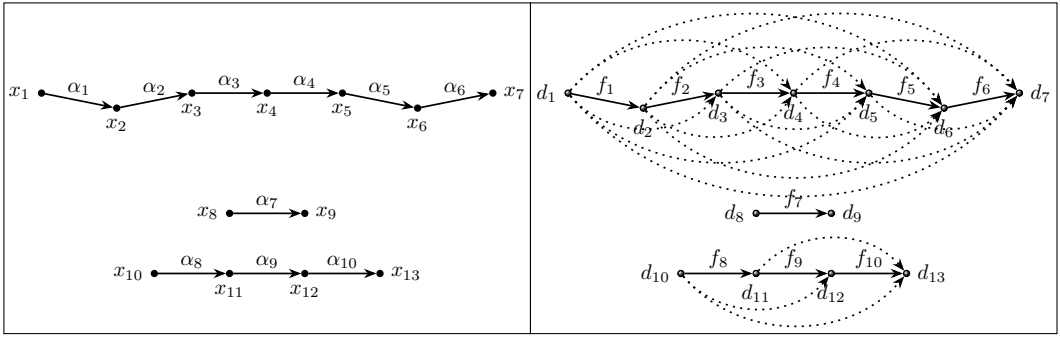
\begin{figure*}[!h]
\centering
\begin{tikzpicture}
	
	\quiverx{alive}{alive}{alive}{alive}{alive}{alive}
	\quivery{alive}
	\quiverz{alive}{alive}{alive}
	\machineexample{0}{12}{-1}{}
	
	\quiverx{alive}{alive}{alive}{alive}{alive}{alive}
	\quivery{alive}
	\quiverz{alive}{alive}{alive}
	\spaceexample{7}{12}{-1}
	
\end{tikzpicture}
\caption{The quiver $Q$ and the corresponding maximal program space in $M$.}
\label{fig:examples-maximal}
\end{figure*}

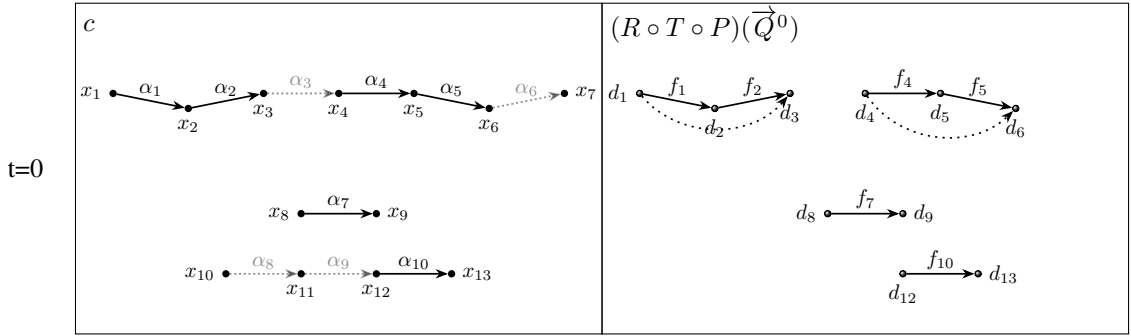
\begin{figure*}[!h]
\centering
\begin{tikzpicture}
	
	\quiverx{alive}{alive}{dead}{alive}{alive}{dead}
	\quivery{alive}
	\quiverz{dead}{dead}{alive}
	\machineexample{0}{12}{0}{$c$}
	
	\quiverx{alive}{alive}{hide}{alive}{alive}{hide}
	\quivery{alive}
	\quiverz{hide}{hide}{alive}
	\spaceexample{7}{12}{0}
	
\end{tikzpicture}
\caption{The initial configuration in both examples and its corresponding program space.}
\label{fig:examples-initial}
\end{figure*}

We now present two examples to demonstrate how (complex) sequential program spaces emerge from simple rules. For all of them, we consider the composition machine $M=(D,F,Q,\mu,S,N,\delta)$ with a fixed set of data types and a fixed set of computons, so $\mu_0$ and $\mu_1$ remain unchanged. We only adjust the local transition functions to demonstrate how different self-organising behaviours can be programmed. In this sense, the initial configuration of $M$ and the neighbourhood collection are shared between the examples. Formally, our composition machine $M$ is defined as follows:

\begin{itemize}
\item $D=\{d_1, d_2, \ldots, d_{13}\}$,
\item $F=\{f_i: d_i \rightarrow d_{i+1} \mid  i \in \mathbb{N} \cap [1,6]~\land~d_i,d_{i+1} \in D\} \cup \{f_7: d_8 \rightarrow d_9 \mid d_8,d_9 \in D\} \cup \{f_i: d_{i+2} \rightarrow d_{i+3} \mid i \in \mathbb{N} \cap [8,10]~\land~d_{i+2},d_{i+3} \in D\}$,
\item $Q=(Q_0,Q_1,s,\tau)$ where:
\begin{itemize}
\item $Q_0=\{x_1,x_2,\ldots,x_{13}\}$,
\item $Q_1=\{\alpha_1,\alpha_2,\ldots,\alpha_{10}\}$,
\item $s=\{(\alpha_i,x_i) \mid  i \in \mathbb{N} \cap [1,6]\} \cup \{(\alpha_7,x_8)\} \cup \{(\alpha_i,x_{i+2}) \mid i \in \mathbb{N} \cap [8,10]\}$, 
\item $\tau=\{(\alpha_i,x_{i+1}) \mid  i \in \mathbb{N} \cap [1,6]\} \cup \{(\alpha_7,x_9)\} \cup \{(\alpha_i,x_{i+3}) \mid i \in \mathbb{N} \cap [8,10]\}$
\end{itemize}
\item $\mu_0=\{(x_i,d_i) \in Q_0 \times D \mid i \in \mathbb{N} \cap [1,13]\}$,
\item $\mu_1=\{(\alpha_i,f_i) \in Q_1 \times F \mid i \in \mathbb{N} \cap [1,10]\}$,
\item $S=\{0,1\}$,
\item $N=\{N_1,N_2,N_3,N_4\}$ where:
\begin{itemize}
\item $N_1=\{(\alpha_7)\}$,
\item $N_2=\{(\alpha_1,\alpha_2),(\alpha_8,\alpha_9)\}$,
\item $N_3=\{(\alpha_5,\alpha_6),(\alpha_9,\alpha_{10})\}$, 
\item $N_4=\{(\alpha_1,\alpha_2,\alpha_3),(\alpha_2,\alpha_3,\alpha_4),(\alpha_3,\alpha_4,\alpha_5),\\(\alpha_4,\alpha_5,\alpha_6),(\alpha_8,\alpha_9,\alpha_{10})\}$
\end{itemize}
\item $\delta_1$, $\delta_2$, $\delta_3$ and $\delta_4$ are defined differently in each example.
\end{itemize}

Figure~\ref{fig:examples-maximal} illustrates the quiver $Q$ and the maximal program space in $M$, viz. $(R \circ T \circ P)(Q)$. According to the Theorem~\ref{theorem:maximal-space}, this space is the most complex that our example machine can construct when all the organisms are alive. For simplicity, we do not show the identity computons looping over each data type, and we do not draw labels on composites to improve readability.

Also, to arise familiarity, all the local transition functions described in this section are based on existing ones. Particularly, isolated organisms evolve through NOT or constant rules. Organisms with a binary neighbourhood change their states via OR or XOR operations. And organisms with a ternary neighbourhood evolve according to \emph{Rule 54} or \emph{Rule 122} (c.f. elementary cellular automata~\cite{wolfram_new_2002}). 

In any case, the synchronous application of the appropriate rules induces a time evolution of the global configuration of $M$, starting from: 

$c=\{(\alpha_1,1),(\alpha_2,1),(\alpha_3,0),(\alpha_4,1),(\alpha_5,1),(\alpha_6,0),\\(\alpha_7,1),(\alpha_8,0),(\alpha_9,0),(\alpha_{10},1)\}$

The initial configuration $c$ and its corresponding program space $(R \circ T \circ P)(\overrightarrow{Q}^0)$ are illustrated in Figure~\ref{fig:examples-initial}. Here, we can see that only two composite computons are available in $(R \circ T \circ P)(\overrightarrow{Q}^0)_1$. Specifically, $f_2 \circ f_1$ is present because $c(\alpha_1)^0=1$, $c(\alpha_2)^0=1$, $\mu_1(\alpha_1)=f_1$ and $\mu_1(\alpha_2)=f_2$. Similarly, we have $f_5 \circ f_4 \in (R \circ T \circ P)(\overrightarrow{Q}^0)_1$ since $c(\alpha_4)^0=1$, $c(\alpha_5)^0=1$, $\mu_1(\alpha_4)=f_4$ and $\mu_1(\alpha_5)=f_5$. 

A glance at Figure~\ref{fig:examples-initial} also reveals that the organisms $\alpha_7$ and $\alpha_{10}$ give rise to the computons $f_7$ and $f_{10}$. This is due to the fact that $c(\alpha_{7})^0=1$, $c(\alpha_{10})^0=1$, $\mu_1(\alpha_7)=f_7$ and $\mu_1(\alpha_{10})=f_{10}$. 

\subsection{Combining the NOT operator, the XOR operator and Rule 54 for the emergence of sequential program spaces}

In this example, the orbit of $M$ is given by the synchronous application of the rules illustrated in Figure~\ref{fig:example1-rules}, starting from the initial configuration $c$ (shown in Figure~\ref{fig:examples-initial}). Particularly, $\delta_1$ is the NOT function, $\delta_2$ and $\delta_3$ are XOR operations, and $\delta_4$ corresponds to \emph{Rule 54}.

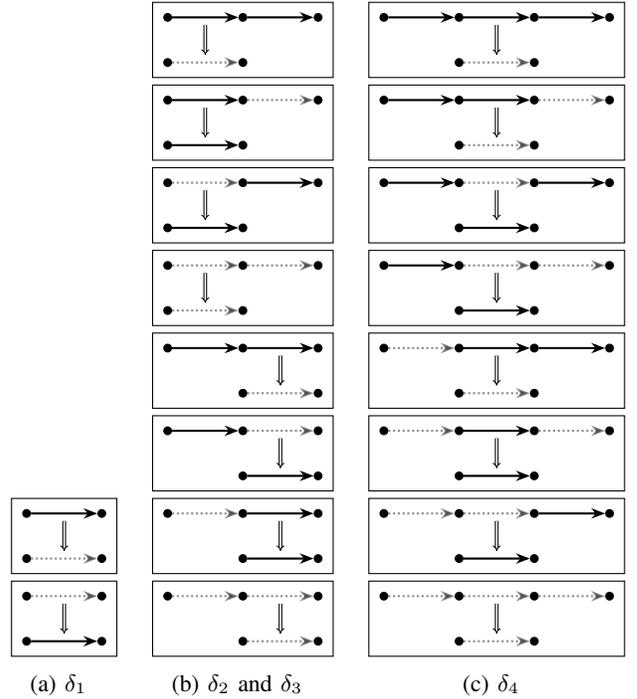
\begin{figure}[!h]
\centering
\subcaptionbox{$\delta_1$ \label{fig:example1-rules-not}}
{
	\begin{tikzpicture}
		\machineruleunary{0}{1.9}{alive}{dead}
		\machineruleunary{0}{0.8}{dead}{alive}
	\end{tikzpicture}
}
\subcaptionbox{$\delta_2$ and $\delta_3$ \label{fig:example1-rules-xor}}
{
	\begin{tikzpicture}
		\machinerulebinary{0}{8.5}{alive}{alive}{dead}{left}
		\machinerulebinary{0}{7.4}{alive}{dead}{alive}{left}	
		\machinerulebinary{0}{6.3}{dead}{alive}{alive}{left}
		\machinerulebinary{0}{5.2}{dead}{dead}{dead}{left}
	
		\machinerulebinary{0}{4.1}{alive}{alive}{dead}{right}
		\machinerulebinary{0}{3}{alive}{dead}{alive}{right}
		\machinerulebinary{0}{1.9}{dead}{alive}{alive}{right}
		\machinerulebinary{0}{0.8}{dead}{dead}{dead}{right}
	\end{tikzpicture}
}
\subcaptionbox{$\delta_4$ \label{fig:example1-rules-54}}
{
	\begin{tikzpicture}
		\machineruleternary{0}{8.5}{alive}{alive}{alive}{dead}	
		\machineruleternary{0}{7.4}{alive}{alive}{dead}{dead}
		\machineruleternary{0}{6.3}{alive}{dead}{alive}{alive}
		\machineruleternary{0}{5.2}{alive}{dead}{dead}{alive}
		\machineruleternary{0}{4.1}{dead}{alive}{alive}{dead}
		\machineruleternary{0}{3}{dead}{alive}{dead}{alive}
		\machineruleternary{0}{1.9}{dead}{dead}{alive}{alive}
		\machineruleternary{0}{0.8}{dead}{dead}{dead}{dead}
	\end{tikzpicture}
}
\caption{Unary neighbourhoods are associated with the NOT operator, binary neighbourhoods with the XOR operator and ternary ones with \emph{Rule 54}.}
\label{fig:example1-rules}
\end{figure}

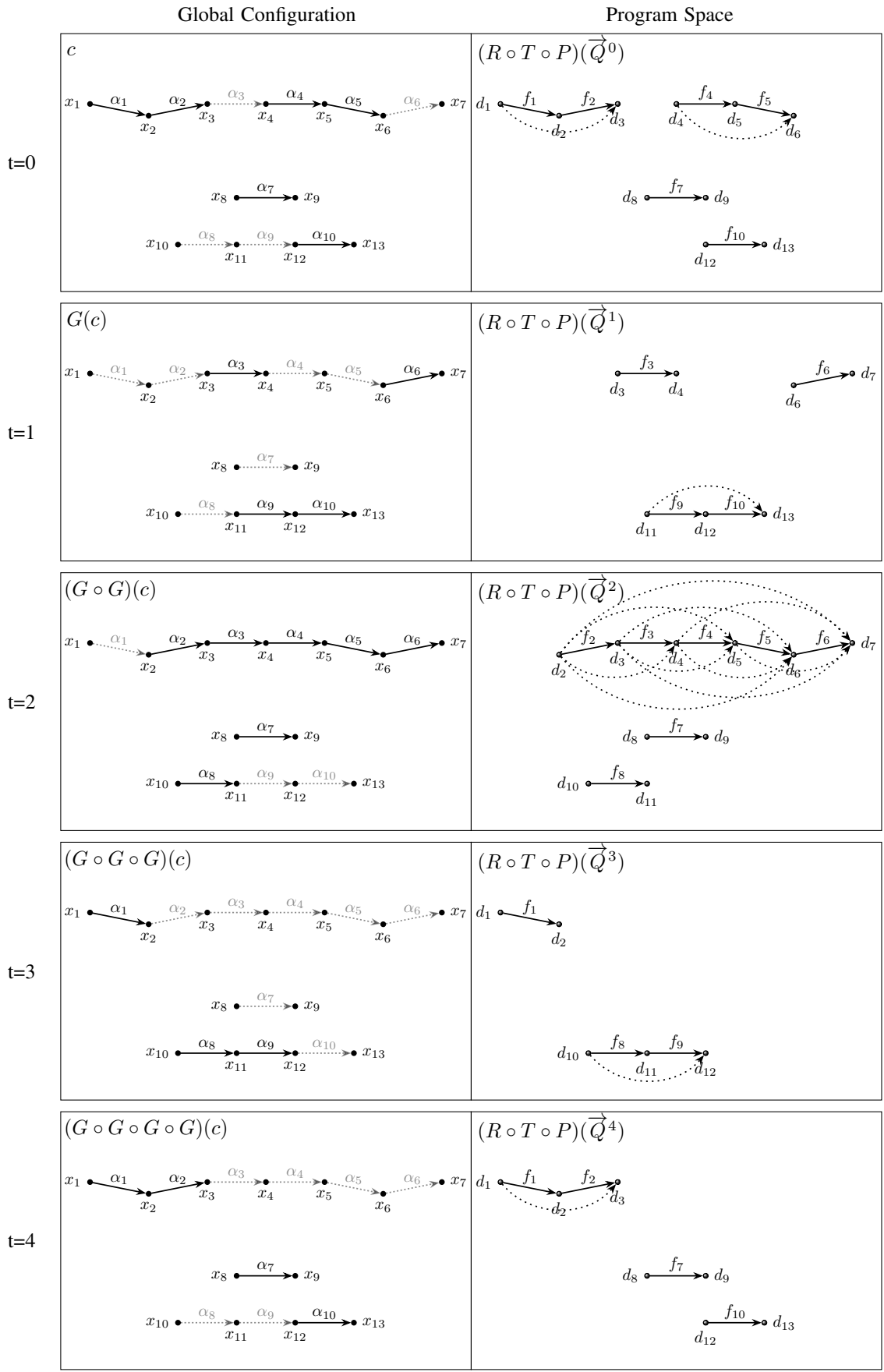
\begin{figure*}
\centering
\begin{tikzpicture}

	\node[text width=5cm] at (4,13.5) {Global Configuration};
	\node[text width=5cm] at (11.3,13.5) {Program Space};
	
	\quiverx{alive}{alive}{dead}{alive}{alive}{dead}
	\quivery{alive}
	\quiverz{dead}{dead}{alive}
	\machineexample{0}{12}{0}{$c$}
	
	\quiverx{dead}{dead}{alive}{dead}{dead}{alive}
	\quivery{dead}
	\quiverz{dead}{alive}{alive}
	\machineexample{0}{7.4}{1}{$G(c)$}
	
	\quiverx{dead}{alive}{alive}{alive}{alive}{alive}
	\quivery{alive}
	\quiverz{alive}{dead}{dead}
	\machineexample{0}{2.8}{2}{$(G \circ G)(c)$}
	
	\quiverx{alive}{dead}{dead}{dead}{dead}{dead}
	\quivery{dead}
	\quiverz{alive}{alive}{dead}
	\machineexample{0}{-1.8}{3}{$(G \circ G \circ G)(c)$}
	
	\quiverx{alive}{alive}{dead}{dead}{dead}{dead}
	\quivery{alive}
	\quiverz{dead}{dead}{alive}
	\machineexample{0}{-6.4}{4}{$(G \circ G \circ G \circ G)(c)$}	
	
	\quiverx{alive}{alive}{hide}{alive}{alive}{hide}
	\quivery{alive}
	\quiverz{hide}{hide}{alive}
	\spaceexample{7}{12}{0}
	
	\quiverx{hide}{hide}{alive}{hide}{hide}{alive}
	\quivery{hide}
	\quiverz{hide}{alive}{alive}
	\spaceexample{7}{7.4}{1}
	
	\quiverx{hide}{alive}{alive}{alive}{alive}{alive}
	\quivery{alive}
	\quiverz{alive}{hide}{hide}
	\spaceexample{7}{2.8}{2}
	
	\quiverx{alive}{hide}{hide}{hide}{hide}{hide}
	\quivery{hide}
	\quiverz{alive}{alive}{hide}
	\spaceexample{7}{-1.8}{3}
	
	\quiverx{alive}{alive}{hide}{hide}{hide}{hide}
	\quivery{alive}
	\quiverz{hide}{hide}{alive}
	\spaceexample{7}{-6.4}{4}	
\end{tikzpicture}
\caption{Orbit of $M$ and program space evolution over five time steps, using the rules described in Figure~\ref{fig:example1-rules}.}
\label{fig:example1-evolution}
\end{figure*}

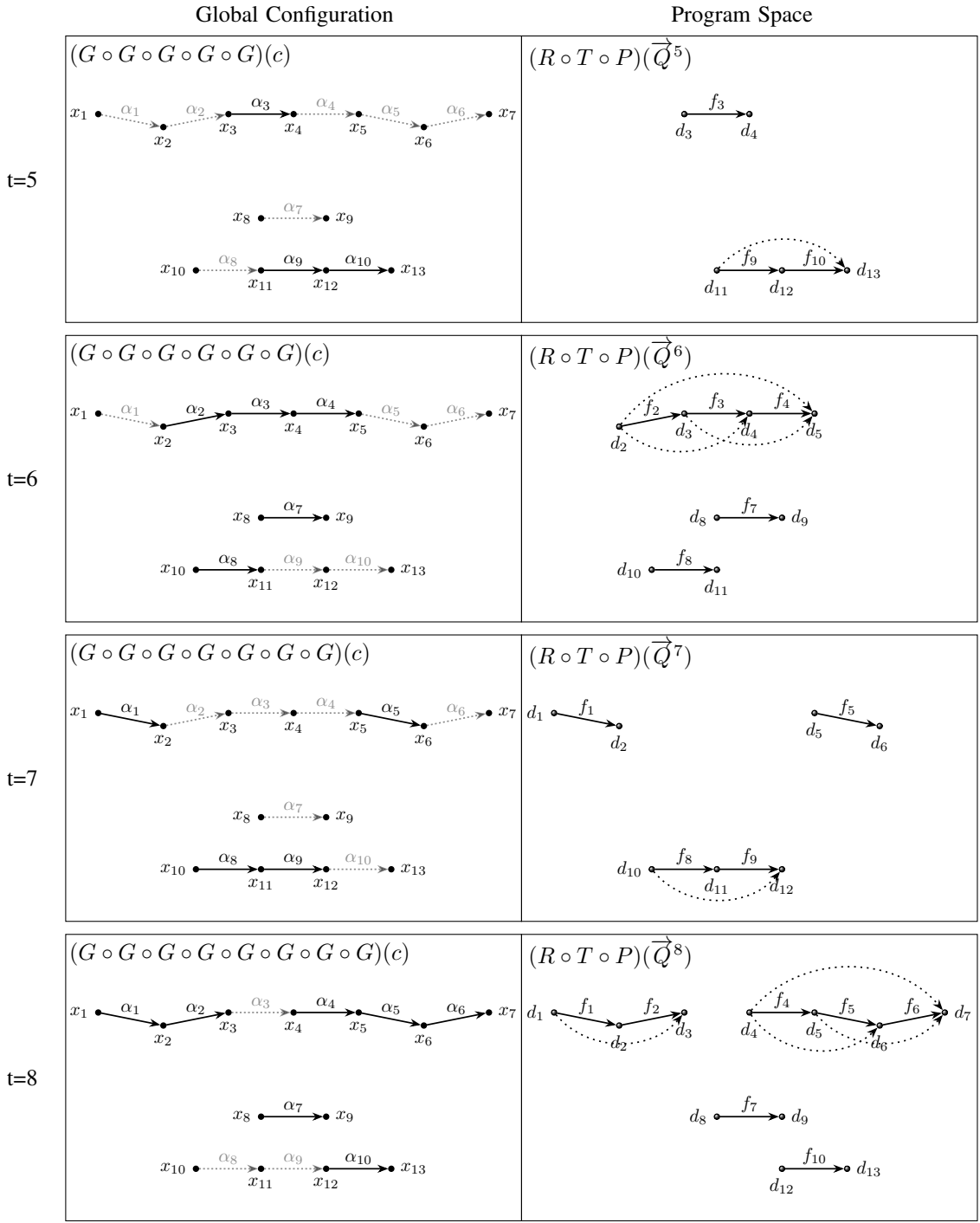
\begin{figure*}
\centering
\begin{tikzpicture}

	\node[text width=5cm] at (4,13.5) {Global Configuration};
	\node[text width=5cm] at (11.3,13.5) {Program Space};
	
	\quiverx{dead}{dead}{alive}{dead}{dead}{dead}
	\quivery{dead}
	\quiverz{dead}{alive}{alive}
	\machineexample{0}{12}{5}{$(G \circ G \circ G \circ G \circ G)(c)$}
	
	\quiverx{dead}{alive}{alive}{alive}{dead}{dead}
	\quivery{alive}
	\quiverz{alive}{dead}{dead}
	\machineexample{0}{7.4}{6}{$(G \circ G \circ G \circ G \circ G \circ G)(c)$}
	
	\quiverx{alive}{dead}{dead}{dead}{alive}{dead}
	\quivery{dead}
	\quiverz{alive}{alive}{dead}
	\machineexample{0}{2.8}{7}{$(G \circ G \circ G \circ G \circ G \circ G \circ G)(c)$}
	
	\quiverx{alive}{alive}{dead}{alive}{alive}{alive}
	\quivery{alive}
	\quiverz{dead}{dead}{alive}
	\machineexample{0}{-1.8}{8}{$(G \circ G \circ G \circ G \circ G \circ G \circ G \circ G)(c)$}
	
	\quiverx{hide}{hide}{alive}{hide}{hide}{hide}
	\quivery{hide}
	\quiverz{hide}{alive}{alive}
	\spaceexample{7}{12}{5}
	
	\quiverx{hide}{alive}{alive}{alive}{hide}{hide}
	\quivery{alive}
	\quiverz{alive}{hide}{hide}
	\spaceexample{7}{7.4}{6}
	
	\quiverx{alive}{hide}{hide}{hide}{alive}{hide}
	\quivery{hide}
	\quiverz{alive}{alive}{hide}
	\spaceexample{7}{2.8}{7}
	
	\quiverx{alive}{alive}{hide}{alive}{alive}{alive}
	\quivery{alive}
	\quiverz{hide}{hide}{alive}
	\spaceexample{7}{-1.8}{8}
\end{tikzpicture}
\caption{Repeating pattern using the initial configuration shown in Figure~\ref{fig:examples-initial} and the rules described in Figure~\ref{fig:example1-rules}.}
\label{fig:example1-repeating}
\end{figure*}

The orbit of $M$ from $t=0$ to $t=4$ is illustrated in Figure~\ref{fig:example1-evolution}, where it is clear that the isolated organism $\alpha_7$ alternates its state via the NOT rule. Accordingly, the computon $f_7$ emerges at even time steps and disappears at odd ones. 

Organisms with a binary neighbourhood evolve via the XOR rule. So, they are alive in the next time step if only one neighbour is alive in the current time step. For example, $\alpha_6$ becomes alive at $t=1$ because its left neighbour $\alpha_5$ is alive at $t=0$. But $\alpha_1$ is dead at $t=1$ since $c(\alpha_1)^0=c(\alpha_2)^0=1$. Consequently, we have $f_6 \in (R \circ T \circ P)(\overrightarrow{Q}^1)_1$ and $f_1 \notin (R \circ T \circ P)(\overrightarrow{Q}^1)_1$. More concretely, the application of the XOR rule on binary neighbourhoods from $t=0$ to $t=1$ results in the following state transitions:

\begin{equation*}
\begin{split}
c(\alpha_1)^1 = \delta_2(c(\alpha_1)^0, c(\alpha_2)^0) = 1 \oplus 1 = 0 \\
c(\alpha_8)^1 = \delta_2(c(\alpha_8)^0, c(\alpha_9)^0) = 0 \oplus 0 = 0 \\
c(\alpha_6)^1 = \delta_3(c(\alpha_5)^0, c(\alpha_6)^0) = 1 \oplus 0 = 1 \\
c(\alpha_{10})^1 = \delta_3(c(\alpha_9)^0, c(\alpha_{10})^0) = 0 \oplus 1 = 1 
\end{split}
\end{equation*}

Organisms with a ternary neighbourhood evolve according to \emph{Rule 54}. For instance, the organism $\alpha_9$ becomes alive at $t=1$ since $c(\alpha_8)^0=c(\alpha_9)^0=0$ and $c(\alpha_{10})^0=1$; thus, making available the computon $f_9 \in (R \circ T \circ P)(\overrightarrow{Q}^1)_1$. As the organism $\alpha_{10}$ is also alive at that moment in time, the composite computon $f_{10} \circ f_9 \in (R \circ T \circ P)(\overrightarrow{Q}^1)_1$ emerges. In general, the application of \emph{Rule 54} on ternary neighbourhoods from $t=0$ to $t=1$ is given by:

\begin{equation*}
\begin{split}
c(\alpha_2)^1 = \delta_4(c(\alpha_1)^0, c(\alpha_2)^0, c(\alpha_3)^0) = \delta_4(1,1,0) = 0 \\
c(\alpha_3)^1 = \delta_4(c(\alpha_2)^0, c(\alpha_3)^0, c(\alpha_4)^0) = \delta_4(1,0,1) = 1 \\
c(\alpha_4)^1 = \delta_4(c(\alpha_3)^0, c(\alpha_4)^0, c(\alpha_5)^0) = \delta_4(0,1,1) = 0 \\
c(\alpha_5)^1 = \delta_4(c(\alpha_4)^0, c(\alpha_5)^0, c(\alpha_6)^0) = \delta_4(1,1,0) = 0 \\
c(\alpha_9)^1 = \delta_4(c(\alpha_8)^0, c(\alpha_9)^0, c(\alpha_{10})^0) = \delta_4(0,0,1) = 1
\end{split}
\end{equation*}

Figure~\ref{fig:example1-evolution} shows the orbit of $M$ over five time steps, where we can see that the most complex space is formed at $t=2$. This is because all the organisms in the path $(\alpha_6,\ldots,\alpha_2)$ become alive. Particularly, $f_6 \circ f_5 \circ f_4 \circ f_3 \circ f_2 \in (R \circ T \circ P)(\overrightarrow{Q}^2)_1$ is the most complex composite computon (i.e., the largest sequential program in our example).

To further explore the emergence of program spaces, we implemented our composition machine $M$.\footnote{The source code of our implementation is publicly available at \\ \url{https://github.com/damianarellanes/compositionmachine}.} After executing this example over $10000$ time steps, we found that there is a repeating pattern every four time steps, starting from $t=5$. This means that the compositions formed before $t=5$ are unique and, after that, there are only four different program spaces. Such a repeating pattern is shown in Figure~\ref{fig:example1-repeating}.

\subsection{Combining the constant operator, the OR operator and Rule 122 for the emergence of sequential program spaces}

In this example, the orbit of $M$ is given by the synchronous application of the rules illustrated in Figure~\ref{fig:example2-rules}, starting from the initial configuration $c$ shown in Figure~\ref{fig:examples-initial}. Here, $\delta_1$ is a constant function, $\delta_2$ and $\delta_3$ are OR operations, and $\delta_4$ corresponds to \emph{Rule 122}.

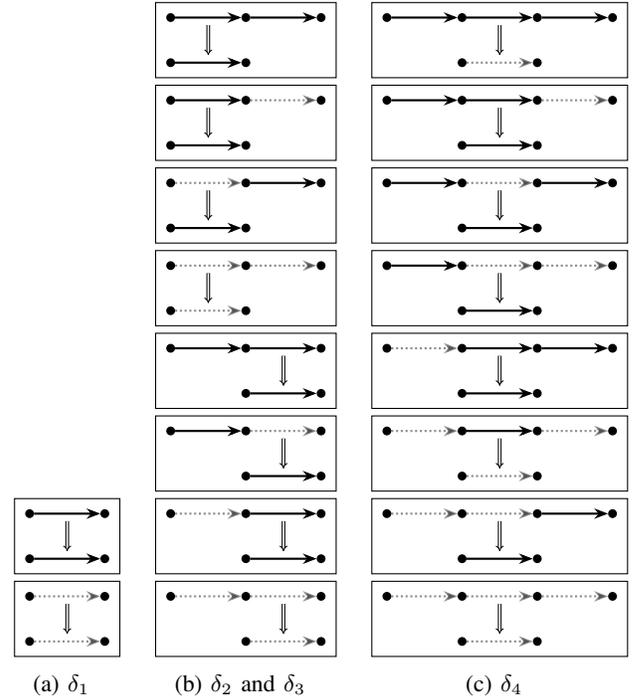
\begin{figure}[!h]
\centering
\subcaptionbox{$\delta_1$ \label{fig:example2-rules-constant}}
{
	\begin{tikzpicture}
		\machineruleunary{0}{1.9}{alive}{alive}
		\machineruleunary{0}{0.8}{dead}{dead}
	\end{tikzpicture}
}
\subcaptionbox{$\delta_2$ and $\delta_3$ \label{fig:example2-rules-or}}
{
	\begin{tikzpicture}
		\machinerulebinary{0}{8.5}{alive}{alive}{alive}{left}
		\machinerulebinary{0}{7.4}{alive}{dead}{alive}{left}	
		\machinerulebinary{0}{6.3}{dead}{alive}{alive}{left}
		\machinerulebinary{0}{5.2}{dead}{dead}{dead}{left}
	
		\machinerulebinary{0}{4.1}{alive}{alive}{alive}{right}
		\machinerulebinary{0}{3}{alive}{dead}{alive}{right}
		\machinerulebinary{0}{1.9}{dead}{alive}{alive}{right}
		\machinerulebinary{0}{0.8}{dead}{dead}{dead}{right}
	\end{tikzpicture}
}
\subcaptionbox{$\delta_4$ \label{fig:example2-rules-122}}
{
	\begin{tikzpicture}
		\machineruleternary{0}{8.5}{alive}{alive}{alive}{dead}	
		\machineruleternary{0}{7.4}{alive}{alive}{dead}{alive}
		\machineruleternary{0}{6.3}{alive}{dead}{alive}{alive}
		\machineruleternary{0}{5.2}{alive}{dead}{dead}{alive}
		\machineruleternary{0}{4.1}{dead}{alive}{alive}{alive}
		\machineruleternary{0}{3}{dead}{alive}{dead}{dead}
		\machineruleternary{0}{1.9}{dead}{dead}{alive}{alive}
		\machineruleternary{0}{0.8}{dead}{dead}{dead}{dead}
	\end{tikzpicture}
}
\caption{Unary neighbourhoods are associated with the constant operator, binary neighbourhoods with the OR operator and ternary ones with \emph{Rule 122}.}
\label{fig:example2-rules}
\end{figure}

\begin{figure*}
\centering
\begin{tikzpicture}

	\node[text width=5cm] at (4,13.5) {Global Configuration};
	\node[text width=5cm] at (11.3,13.5) {Program Space};	
	
	\quiverx{alive}{alive}{dead}{alive}{alive}{dead}
	\quivery{alive}
	\quiverz{dead}{dead}{alive}
	\machineexample{0}{12}{0}{$c$}
	
	\quiverx{alive}{alive}{alive}{alive}{alive}{alive}
	\quivery{alive}
	\quiverz{dead}{alive}{alive}
	\machineexample{0}{7.4}{1}{$G(c)$}
	
	\quiverx{alive}{dead}{dead}{dead}{dead}{alive}
	\quivery{alive}
	\quiverz{alive}{alive}{alive}
	\machineexample{0}{2.8}{2}{$(G \circ G)(c)$}
	
	\quiverx{alive}{alive}{dead}{dead}{alive}{alive}
	\quivery{alive}
	\quiverz{alive}{dead}{alive}
	\machineexample{0}{-1.8}{3}{$(G \circ G \circ G)(c)$}
	
	\quiverx{alive}{alive}{alive}{alive}{alive}{alive}
	\quivery{alive}
	\quiverz{alive}{alive}{alive}
	\machineexample{0}{-6.4}{4}{$(G \circ G \circ G \circ G)(c)$}	
	
	\quiverx{alive}{alive}{hide}{alive}{alive}{hide}
	\quivery{alive}
	\quiverz{hide}{hide}{alive}
	\spaceexample{7}{12}{0}
	
	\quiverx{alive}{alive}{alive}{alive}{alive}{alive}
	\quivery{alive}
	\quiverz{hide}{alive}{alive}
	\spaceexample{7}{7.4}{1}
	
	\quiverx{alive}{hide}{hide}{hide}{hide}{alive}
	\quivery{alive}
	\quiverz{alive}{alive}{alive}
	\spaceexample{7}{2.8}{2}
	
	\quiverx{alive}{alive}{hide}{hide}{alive}{alive}
	\quivery{alive}
	\quiverz{alive}{hide}{alive}
	\spaceexample{7}{-1.8}{3}
	
	\quiverx{alive}{alive}{alive}{alive}{alive}{alive}
	\quivery{alive}
	\quiverz{alive}{alive}{alive}
	\spaceexample{7}{-6.4}{4}	
\end{tikzpicture}
\end{figure*}

\begin{figure*}
\centering
\begin{tikzpicture}	
	
	\quiverx{alive}{dead}{dead}{dead}{dead}{alive}
	\quivery{alive}
	\quiverz{alive}{dead}{alive}
	\machineexample{0}{12}{5}{$(G \circ G \circ G \circ G \circ G)(c)$}
	
	\quiverx{alive}{alive}{dead}{dead}{alive}{alive}
	\quivery{alive}
	\quiverz{alive}{alive}{alive}
	\machineexample{0}{7.4}{6}{$(G \circ G \circ G \circ G \circ G \circ G)(c)$}
	
	\quiverx{alive}{alive}{alive}{alive}{alive}{alive}
	\quivery{alive}
	\quiverz{alive}{dead}{alive}
	\machineexample{0}{2.8}{7}{$(G \circ G \circ G \circ G \circ G \circ G \circ G)(c)$}
	
	\quiverx{alive}{dead}{dead}{dead}{dead}{alive}
	\quivery{alive}
	\quiverz{alive}{alive}{alive}
	\machineexample{0}{-1.8}{8}{$(G \circ G \circ G \circ G \circ G \circ G \circ G \circ G)(c)$}
	
	\quiverx{alive}{alive}{dead}{dead}{alive}{alive}
	\quivery{alive}
	\quiverz{alive}{dead}{alive}
	\machineexample{0}{-6.4}{9}{$(G \circ G \circ G \circ G \circ G \circ G \circ G \circ G \circ G)(c)$}	
	
	\quiverx{alive}{hide}{hide}{hide}{hide}{alive}
	\quivery{alive}
	\quiverz{alive}{hide}{alive}
	\spaceexample{7}{12}{5}
	
	\quiverx{alive}{alive}{hide}{hide}{alive}{alive}
	\quivery{alive}
	\quiverz{alive}{alive}{alive}
	\spaceexample{7}{7.4}{6}
	
	\quiverx{alive}{alive}{alive}{alive}{alive}{alive}
	\quivery{alive}
	\quiverz{alive}{hide}{alive}
	\spaceexample{7}{2.8}{7}
	
	\quiverx{alive}{hide}{hide}{hide}{hide}{alive}
	\quivery{alive}
	\quiverz{alive}{alive}{alive}
	\spaceexample{7}{-1.8}{8}
	
	\quiverx{alive}{alive}{hide}{hide}{alive}{alive}
	\quivery{alive}
	\quiverz{alive}{hide}{alive}
	\spaceexample{7}{-6.4}{9}
\end{tikzpicture}
\caption{Orbit of $M$ and program space evolution over ten time steps, using the rules described in Figure~\ref{fig:example2-rules}.}
\label{fig:example2-evolution}
\end{figure*}

The orbit of $M$ from $t=0$ to $t=9$ is illustrated in Figure~\ref{fig:example2-evolution}, where it is clear that the isolated organism $\alpha_7$ remains alive at every moment. Accordingly, the computon $f_7$ consistently appears in all the program spaces. 

Organisms with a binary neighbourhood evolve via the OR rule. So, they are alive in the next time step if at least one neighbour is alive in the current time step. For example, $\alpha_6$ becomes alive at $t=1$ because its left neighbour $\alpha_5$ is alive at $t=0$. But $\alpha_8$ remains dead at $t=1$ since $c(\alpha_8)^0=c(\alpha_9)^0=0$. Consequently, we have $f_6 \in (R \circ T \circ P)(\overrightarrow{Q}^1)_1$ and $f_8 \notin (R \circ T \circ P)(\overrightarrow{Q}^1)_1$. More concretely, the application of the OR rule on binary neighbourhoods from $t=0$ to $t=1$ results in the following state transitions:

\begin{equation*}
\begin{split}
c(\alpha_1)^1 = \delta_2(c(\alpha_1)^0, c(\alpha_2)^0) = 1 \lor 1 = 1 \\
c(\alpha_8)^1 = \delta_2(c(\alpha_8)^0, c(\alpha_9)^0) = 0 \lor 0 = 0 \\
c(\alpha_6)^1 = \delta_3(c(\alpha_5)^0, c(\alpha_6)^0) = 1 \lor 0 = 1 \\
c(\alpha_{10})^1 = \delta_3(c(\alpha_9)^0, c(\alpha_{10})^0) = 0 \lor 1 = 1 
\end{split}
\end{equation*}

Organisms with a ternary neighbourhood evolve according to \emph{Rule 122}. For instance, the organism $\alpha_3$ becomes alive at $t=1$ since $c(\alpha_2)^0=c(\alpha_4)^0=1$ and $c(\alpha_3)^0=0$; thus, making available the computon $f_3 \in (R \circ T \circ P)(\overrightarrow{Q}^1)_1$. As all the organisms in the path $(\alpha_6,\ldots,\alpha_2,\alpha_1)$ are alive at $t=1$, the space $(R \circ T \circ P)(\overrightarrow{Q}^1)$ defines all the possible compositions that can be formed with the computons $f_1,\ldots,f_6$. In general, the application of \emph{Rule 122} on ternary neighbourhoods from $t=0$ to $t=1$ is given by:

\begin{equation*}
\begin{split}
c(\alpha_2)^1 = \delta_4(c(\alpha_1)^0, c(\alpha_2)^0, c(\alpha_3)^0) = \delta_4(1,1,0) = 1 \\
c(\alpha_3)^1 = \delta_4(c(\alpha_2)^0, c(\alpha_3)^0, c(\alpha_4)^0) = \delta_4(1,0,1) = 1 \\
c(\alpha_4)^1 = \delta_4(c(\alpha_3)^0, c(\alpha_4)^0, c(\alpha_5)^0) = \delta_4(0,1,1) = 1 \\
c(\alpha_5)^1 = \delta_4(c(\alpha_4)^0, c(\alpha_5)^0, c(\alpha_6)^0) = \delta_4(1,1,0) = 1 \\
c(\alpha_9)^1 = \delta_4(c(\alpha_8)^0, c(\alpha_9)^0, c(\alpha_{10})^0) = \delta_4(0,0,1) = 1
\end{split}
\end{equation*}

Figure~\ref{fig:example2-evolution} shows the orbit of $M$ over ten time steps, where we can see that the most complex space is formed at $t=4$. This emergent space is indeed $(R \circ T \circ P)(Q)$. With the help of our implemented tool, we executed our example over $10000$ time steps and we found that the maximal program space $(R \circ T \circ P)(Q)$ is part of a repeating pattern whose period corresponds to six time steps (see Figure~\ref{fig:example2-evolution} from $t=2$ to $t=7$). Particularly, the compositions emerging before $t=2$ are unique and, after that, there are only six different program spaces. 

\section{Conclusions}
\label{sec:conclusions}

In this paper, we introduced the notion of composition machines for the emergence of sequential program spaces via self-organisation rules. A composition machine evolves a quiver in discrete-time and its global state is semantically equivalent to a category that defines a space of all the possible compositions of computons at some moment in time. To demonstrate its operation, we presented two examples in which complex sequential program spaces emerge from simple rules. 

As it is not trivial to engineer emergence, we plan to investigate novel mechanisms for defining rules that match contextual intention. Also, we plan to study alternative neighbourhoods to allow the evolution of cyclic quivers and, thus, the emergence of infinite categories. In this vein, we would like to explore different program space patterns through the synchronous application of varied local transition rules. In any case, Category Theory will be the \emph{de facto} reasoning framework.

We believe that Category Theory will play a fundamental role in the understanding of self-organisation and, in particular, in the study of self-organising software models. This is because such a theory follows an emergentist view rather than a reductionist one. So, the main focus of Category Theory is not to study isolated abstract objects but the relevant interactions between them and their composition. It is indeed a formal framework to reasoning about compositionality and self-organisation. 

This paper brings together these worlds in the form of an abstract machine that allows the definition and the execution of self-organising software models. The aim of this machine is not to compute single programs, but to build programs (i.e., composite computons) from other programs (i.e., predefined computons) via self-organisation rules. Under this umbrella, we envision that in the future, instead of manually coding them, complex software models will grow from simple rules (just as biological organisms do).

\bibliographystyle{IEEEtran}
\bibliography{IEEEabrv,refs}

\end{document}